\documentclass[aos,preprint]{imsart}
\usepackage{amsthm,amsmath}
\usepackage[colorlinks,citecolor=blue,urlcolor=blue]{hyperref}

\usepackage[margin=1in]{geometry}
\usepackage{bm}
\usepackage{graphicx}
\usepackage{amssymb,amsmath,amsthm,mathrsfs}
\usepackage{epstopdf}
\usepackage{algorithm}
\usepackage{amsfonts,delarray}
\usepackage{algorithm,algcompatible,amsmath}
\usepackage{rotating,color}
\usepackage{subfigure}
\usepackage{multirow}
\usepackage{titlesec,textcase}
\usepackage{longtable}
\usepackage{bbm}
\usepackage{rotating}
\newtheorem{Theorem}{Theorem}[section]
\newtheorem{Lemma}[Theorem]{Lemma}
\newtheorem{Corollary}[Theorem]{Corollary}
\newtheorem{Proposition}[Theorem]{Proposition}

\newtheorem{Remark}{Remark}[section]
\theoremstyle{definition}

\newcommand{\bbN}{{\mathbb{N}}}
\newcommand{\E}{{\mathbb{E}}}
\RequirePackage[normalem]{ulem} %DIF PREAMBLE
\definecolor{rp}{RGB}{83,54,106}

\def\boxit#1{\vbox{\hrule\hbox{\vrule\kern6pt\vbox{\kern6pt#1\kern6pt}\kern6pt\vrule}\hrule}}

\begin{document}

\begin{frontmatter}
	\title{A Likelihood-Ratio Type Test for Stochastic Block Models with Bounded Degrees}
	\runtitle{Testing Sparse Network}
	
	\begin{aug}
		\author{\fnms{Mingao } \snm{Yuan}\thanksref{m1}\ead[label=e1]{mingyuan@iupui.edu}},
		\author{\fnms{Yang} \snm{Feng}\thanksref{m2,t1}\ead[label=e2]{yang.feng@columbia.edu}}
		\and
		\author{\fnms{Zuofeng} \snm{Shang}\thanksref{m1,t2}
			\ead[label=e3]{shangzf@iu.edu}}

		\thankstext{t1}{Supported by NSF CAREER Grant DMS-1554804.}
		
		\thankstext{t2}{Corresponding author. Supported by 
			NSF DMS-1764280 and 
			a startup grant from IUPUI.}
		\runauthor{Yuan, Feng and Shang}

		\affiliation{IUPUI\thanksmark{m1} and Columbia University\thanksmark{m2}}
		
		\address{Department of Mathematical Sciences,\\
			Indiana University-Purdue University-Indianapolis,\\
			402 N Blackford St,\\
			Indianapolis, IN 46202\\
			\printead{e1}\\
			\phantom{E-mail:\ }}
		
		\address{Department of Statistics, \\
			Columbia University,\\
			1255 Amsterdam Ave,\\
			New York, NY 10027\\
			\printead{e2}\\}
		%\printead{u1}}
		\address{Department of Mathematical Sciences,\\
			Indiana University-Purdue University-Indianapolis,\\
			402 N Blackford St,\\
			Indianapolis, IN 46202\\
			\printead{e3}\\
			\phantom{E-mail:\ }}
	\end{aug}
	
	\begin{abstract}
		A fundamental problem in network data analysis is to test Erd\"{o}s-R\'{e}nyi model $\mathcal{G}\left(n,\frac{a+b}{2n}\right)$ versus a bisection stochastic block model
		$\mathcal{G}\left(n,\frac{a}{n},\frac{b}{n}\right)$,
		where $a,b>0$ are constants that represent the expected degrees of the graphs and $n$ denotes the number of nodes.
		This problem serves as the foundation of many 
		other problems such as testing-based methods for determining the number of communities
		(\cite{BS16,L16}) and community detection (\cite{MS16}).
		Existing work has been focusing on growing-degree regime $a,b\to\infty$ (\cite{BS16,L16,MS16,BM17,B18,GL17a,GL17b})
		while leaving the bounded-degree regime untreated. In this paper, we propose a likelihood-ratio (LR) type procedure based on regularization
		to test stochastic block models with bounded degrees.
		We derive the limit distributions as power Poisson laws under both null and alternative hypotheses,
		based on which the limit power of the test is carefully analyzed.
		We also examine a Monte-Carlo method that partly resolves the computational cost issue.
		The proposed procedures are examined by both simulated and real-world data.
		The proof depends on a contiguity theory developed by Janson \cite{J95}. 
	\end{abstract}
	
	
\end{frontmatter}

\section{Introduction}
In recent years, stochastic block model (SBM) has attracted increasing attention in statistics and machine learning.
It provides the researchers a ground to study many important problems that arise in network data
such as community detection or clustering (\cite{ACBL13,AL18,NN14,SB15,BC09,ZLZ12}), 
goodness-of-fit of SBMs (\cite{BS16,L16,MS16,BM17,B18,GL17a,GL17b}) or various phase transition phenomena
(\cite{MNS15,MNS17,AS17}). See \cite{A17} for a comprehensive review about recent development in this field.
A key assumption in most of the literature is that the expected degree of every node 
tends to infinity along with the number of nodes $n$.
For instance, in community detection (\cite{BC09,ZLZ12}), such a condition is needed for proving weak consistency
of the detection methods;
to prove strong consistency, the expected degree is further assumed to grow faster than $\log{n}$.
For goodness-of-fit test, the growing-degree condition is needed to derive various
asymptotic distributions for the test statistics (\cite{BS16,L16,BM17,B18,GL17a,GL17b}).

Many real-world network data sets are highly sparse.
For instance, the LinkedIn network, the real-world
coauthorship networks, power transmission networks and web link networks all have small average
degrees (see \cite{LLDM08, S01}). Therefore, it is reasonable to assume bounded degrees in such networks.
There is a breakthrough recently made by \cite{MNS15,MNS17,AS17} about the possibility
of successfully detecting the community structures when the expected degree of SBM is bounded.
Specifically, the signal-to-noise ratio (SNR) of the multi-community SBM is 
used in these work as a phase transition parameter to indicate the possibility of successful detection.
Motivated by such a groundbreaking result, it is natural to ask whether one can 
propose successful testing methods for SBMs with bounded degrees.
Progress in this field may help researchers better understand the roles played by the expected degrees
of SBMs in hypothesis testing, as well as provide a substantially broader scope of network models in which a successful test is possible. 

In this paper, we address this problem in the bisection SBM scenario.
We propose a likelihood-ratio (LR) type test statistic
to distinguish an Erd\"{o}s-R\'{e}nyi model versus a bisection SBM
whose expected degrees are finite constants, and investigate its asymptotic properties.
In what follows, we describe the models and our contributions more explicitly.

\subsection{Models and Our Contributions.}
Let us provide a brief review for Erd\"{o}s-R\'{e}nyi model and bisection SBM.
Throughout the whole paper, assume that $a>b>0$ are \textit{fixed and known} constants unless otherwise indicated. 
For $n\in\bbN$, let 
$\mathcal{G}\left(n,\frac{a}{n},\frac{b}{n}\right)$ denote the bisection stochastic block model
of random $\pm$-labeled graphs in which each vertex 
$u\in[n]:=\{1,2,\ldots,n\}$ is assigned, independently and uniformly at random, a label $\sigma_u\in\{\pm\}$,
and then each possible edge $(u,v)$ is included with probability $a/n$ if $\sigma_u=\sigma_v$
and with probability $b/n$ if $\sigma_u\neq\sigma_v$. 
Let $A=[A_{uv}]_{u,v=1}^n\in\{0,1\}^{n\times n}$ denote the observed symmetric adjacency matrix in which $A_{uu}=0$ for all $1\le u\le n$,
and for $1\le u<v\le n$,
$A_{uv}=1$ indicates the inclusion of edge $(u,v)$ and $A_{uv}=0$ otherwise.
Conditional on
$\sigma=(\sigma_1,\ldots,\sigma_n)$,
the variables $A_{uv}$, $1\le u<v\le n$, are assumed to be independent which follow 
\begin{eqnarray}\label{G:p:q:model}
\textrm{$P(A_{uv}=1|\sigma)=p_{uv}(\sigma)$ and $P(A_{uv}=0|\sigma)=q_{uv}(\sigma)$,}
\end{eqnarray}
where
\[
p_{uv}(\sigma)=\left\{\begin{array}{cc}
\frac{a}{n},&\sigma_u=\sigma_v\\
\frac{b}{n},&\sigma_u\neq\sigma_v
\end{array}\right.,
q_{uv}(\sigma)=1-p_{uv}(\sigma).
\]
The Erd\"{o}s-R\'{e}nyi model $\mathcal{G}\left(n,\frac{a+b}{2n}\right)$ has the same average degree as $\mathcal{G}\left(n,\frac{a}{n},\frac{b}{n}\right)$. 
It is interesting to decide which model an observed graph is generated from.
Specifically, we are interested in the following hypothesis testing problem
\begin{equation}\label{H0:H1}
\textrm{$H_0$: $A\sim \mathcal{G}\left(n,\frac{a+b}{2n}\right)$\,\,\,\, 
vs.\,\,\,\, $H_1:$ $A\sim \mathcal{G}\left(n,\frac{a}{n},\frac{b}{n}\right)$.}
\end{equation}
To be more specific, we want to test whether the 
nodes on an observed random graph belong to the same community, or
they belong to two equal-sized communities.

Let $\kappa=\frac{(a-b)^2}{2(a+b)}$ denote the signal-to-noise ratio (SNR) associated with $\mathcal{G}\left(n,\frac{a}{n},\frac{b}{n}\right)$.
It was conjectured by Decelle, Krzkala, Moore and Zdeborov\'{a} (\cite{DKMZ11}) 
that successful community detection is possible when $\kappa\ge1$,
and impossible when $\kappa<1$.
This conjecture was recently proved by Mossel, Neeman and Sly (\cite{MNS15}) through Janson's continuity theory (\cite{J95}).
In the meantime, their result indicates that \textit{no test can be successful when $\kappa<1$}
(see \cite{MNS15,MS16}), and so we primarily focus on the high SNR scenario $\kappa\ge1$.
Classic likelihood-ratio (LR) tests for (\ref{H0:H1}) are not valid since the probability measures associated with
$H_0$ and $H_1$ are asymptotically orthogonal as discovered by \cite{MNS15}.
The result of \cite{MNS15} also implies that counting the cycles of length $\log^{1/4}{n}$ leads to an asymptotically 
valid test; see their Theorem 4. However, such test is unrealistic since $n$ should be at least $e^{81}$ to make the length at least 3.
In Section \ref{sec:varepsilon:LR}, we propose a regularized LR-type test for (\ref{H0:H1}) to address these limitations.
Our test does not suffer from the orthogonality issue of LR and is applicable for moderately large $n$.
Our test involves a regularization parameter that can reduce the variability of the classic LR test so that it becomes valid.
Based on a contiguity theory for random regular graphs developed by Janson \cite{J95},
we derive the asymptotic distributions as power Poisson laws under both $H_0$ and $H_1$, which turn
out to be infinite products of power Poisson variables (see Section \ref{sec:ppl}). 
Based on power Poisson laws, we rigorously analyze the asymptotic power of our test. In Section \ref{sec:power:analysis}, we show that
the test is powerful provided that $\kappa$ approaches infinity,
and the limit power is not sensitive to the choice of regularization parameter.
Our test is practically useful in that the parameters $a,b$ can be consistently estimated when $\kappa>1$, and
so the regularization parameter can be empirically selected. Our procedure is based on averaged likelihood-ratios whose computational cost scales exponentially with
$n$. This computational issue is partly resolved in Section \ref{sec:approx} via Monte Carlo approximations,
with the number of experiments suggested to guarantee the success of such approximations.
Simulation examples are provided in Section \ref{sec:sim} to demonstrate the finite sample performance of our methods.
In particular, our method achieves desirable size and power, while the methods designed for denser graphs
appear to be less powerful. 

\subsection{Related References.}

The problem of testing (\ref{H0:H1}) has been recently considered by \cite{BS16,L16,MS16,BM17,B18,GL17a,GL17b} but only
in the growing-degree regime, i.e., $a,b\to\infty$. 
Specifically, \cite{BS16,MS16,L16} proposed spectral algorithms;
\cite{BM17,B18} proposed linear spectral statistics and LR test relating to signed cycles;
\cite{GL17a,GL17b} proposed algorithms based subgraph counts. 
In particular, the LR test by \cite{BM17} was proposed under low SNR
which may not be directly applicable here. 
The growing-degree condition is necessary to guarantee the validity of all these methods
which also result in different asymptotic laws than ours.
As far as we know, an effective testing procedure that distinguishes SBMs with bounded degrees is still missing. 
As a side remark, the power Poisson law is unique in sparse network models with bounded degrees as demonstrated in \cite{J95}.
In the end, we mention a few papers addressing different models or testing problems than ours:
\cite{FH15} proposed a test for examining dependence between network factors and nodal-level attributes;
\cite{MPOW17} proposed a variant of multivariate t-test 
for model diagnosis based on a collection of network samples.

\section{LR-Type Test and Asymptotic Properties}\label{sec:varepsilon:LR}
The classic LR test requires the calculations of 
the marginal probability distributions of $A_{uv}$'s under both $H_0$ and $H_1$.
By straightforward calculations, it can be shown that,
under $H_1$, the marginal distribution of $A$ is
\[
P_1(A)=\sum_{\sigma\in\{\pm\}^n}P(A|\sigma)P(\sigma)=2^{-n}\sum_{\sigma\in\{\pm\}^n}\prod_{u<v}p_{uv}(\sigma)^{A_{uv}}q_{uv}(\sigma)^{1-A_{uv}};
\]
and under $H_0$, the marginal distribution of $A$ is
\[
P_0(A)=\prod_{u<v}p_0^{A_{uv}}q_0^{1-A_{uv}},
\]
where $p_0=1-q_0=\frac{a+b}{2n}$.
The classic LR test for (\ref{H0:H1}) is then given as follows:
\begin{equation}\label{lrt:case1}
Y_n=\frac{P_1(A)}{P_0(A)}=2^{-n}\sum_{\sigma\in\{\pm\}^n}\prod_{u<v}\left(\frac{p_{uv}(\sigma)}{p_0}\right)^{A_{uv}}
\left(\frac{q_{uv}(\sigma)}{q_0}\right)^{1-A_{uv}},
\end{equation}
where $p_{uv}(\sigma)$ and $q_{uv}(\sigma)$ are defined in (\ref{G:p:q:model}).
However, \cite{MNS15} shows that $P_0(\cdot)$ and $P_1(\cdot)$ are
asymptotically orthogonal when $\kappa\ge1$. So with positive probability,
$Y_n$ is asymptotically degenerate to either $0$ or $\infty$.
Here we provide a more heuristic understanding for such degenerateness phenomenon.
Note that the probability ratio $\frac{p_{uv}(\sigma)}{p_0}$ 
is equal to either $\frac{2a}{a+b}$ or $\frac{2b}{a+b}$,
depending on whether $u,v$ belong to the same community.
When $\kappa\ge1$, i.e., $a-b$ is large compared with $a+b$,
the two probability ratios considerably differ from each other which brings
too much uncertainty into $Y_n$.

We propose a regularized LR test, called as $\varepsilon$-LR test,
to resolve the degenerateness issue. The idea is quite natural: incorporate a regularization parameter $\varepsilon$ into $Y_n$ 
to reduce its uncertainty. 
Our $\varepsilon$-LR test
is defined as follows. 
Let $\kappa_\varepsilon=\frac{(a_\varepsilon-b_\varepsilon)^2}{2(a+b)}$, where $a_\varepsilon=a-\varepsilon$, $b_\varepsilon=b+\varepsilon$.
For any $\varepsilon$ satisfying
\begin{equation}\label{vareps:condition}
\textrm{$0<\varepsilon<\frac{a-b}{2}$ and $\kappa_\varepsilon<1$,}
\end{equation}
define
\begin{equation}\label{lrt:case2}
Y_n^\varepsilon=2^{-n}\sum_{\sigma\in\{\pm\}^n}\prod_{u<v}\left(\frac{p_{uv}^\varepsilon(\sigma)}{p_0}\right)^{A_{uv}}
\left(\frac{q_{uv}^\varepsilon(\sigma)}{q_0}\right)^{1-A_{uv}},
\end{equation}
where 
\[
p_{uv}^\varepsilon(\sigma)=\left\{\begin{array}{cc}
\frac{a_\varepsilon}{n},&\sigma_u=\sigma_v\\
\frac{b_\varepsilon}{n},&\sigma_u\neq\sigma_v
\end{array}\right.,\,\,\,\,\,\,
q_{uv}^\varepsilon(\sigma)=1-p_{uv}^\varepsilon(\sigma).
\]
In other words, we replace $p_{uv}(\sigma)$ and $q_{uv}(\sigma)$
in (\ref{lrt:case1}) by their counterparts $p_{uv}^\varepsilon(\sigma)$ and $q_{uv}^\varepsilon(\sigma)$.
The new probability ratio $\frac{p_{uv}^\varepsilon(\sigma)}{p_0}$
is equal to either $\frac{2a_\varepsilon}{a+b}$ or $\frac{2b_\varepsilon}{a+b}$,
which are closer to each other due to regularization.
Such a trick will be proven to effectively reduce the variability of the classic LR test.
Asymptotic distributions and power analysis of $Y_n^\varepsilon$ are provided in 
subsequent Sections \ref{sec:ppl} and \ref{sec:power:analysis}.

\begin{Remark}
A more naive approach is to
reject $H_0$ if $Y_n>c$ with $c>0$ a predetermined constant.
However, the choice of $c$ is a challenging issue. In particular, due to the degenerateness of $Y_n$,
it is hard to determine the (asymptotic) probability of rejection
given any value of $c$, which poses challenges in analyzing size and power of the test.
Instead, our $\varepsilon$-LR test has valid asymptotic distributions 
which avoids the above issues.
\end{Remark}

\subsection{Power Poisson Laws.}\label{sec:ppl}
Let us first present a power Poisson law for $Y_n^\varepsilon$ under $H_0$.

\begin{Theorem}\label{testing:consistency:case2}
If $\kappa\ge1$ and $\varepsilon$ satisfies (\ref{vareps:condition}),
then under $H_0$, $Y_n^\varepsilon\overset{d}{\to}W_0^\varepsilon$ as $n\to\infty$, where
\[
W_0^\varepsilon=\prod_{m=3}^\infty\left(1+\delta_m^\varepsilon\right)^{Z_m^0}\exp\left(-\lambda_m\delta_m^\varepsilon\right),\,\,\,\,
Z_m^0\overset{ind}{\sim}\textrm{Poisson}\left(\lambda_m\right).
\]
Here, $\lambda_m=\frac{1}{2m}\left(\frac{a+b}{2}\right)^m$
and $\delta_m^\varepsilon=\left(\frac{a_\varepsilon-b_\varepsilon}{a+b}\right)^m$.
\end{Theorem}
Theorem \ref{testing:consistency:case2} shows that, under $H_0$, $Y_n^\varepsilon$
converges in distribution to an infinite product of power Poisson variables.
Its proof is based on a contiguity theory for regular random graphs developed by \cite{J95}.
Power Poisson law is unique in sparse network with bounded degree,
e.g., the number of subgraphs, the number of perfect matchings and 
the number of edge colourings all follow such a law (see \cite{J95}).
This decidedly differs from the growing-degree regime. For instance, when the average degree
is growing along with $n$, \cite{BS16} proposed a spectral algorithm that follows Tracy-Widom law;
\cite{BM17,B18} examined the classic LR statistics under $\kappa<1$ and linear spectral statistics relating to signed cycles that follow power Gaussian law; 
\cite{GL17a,GL17b} proposed subgraph-based algorithms that follow Gaussian distributions.  

According to Theorem \ref{testing:consistency:case2}, we test (\ref{H0:H1}) at significance level $\alpha$
based on the following rule:
\[
\textrm{reject $H_0$ iff $Y_n^\varepsilon\ge w_\alpha^\varepsilon$,}
\]
where $w_\alpha^\varepsilon>0$ satisfies $P(W_0^\varepsilon\le w_\alpha^\varepsilon)=1-\alpha$.

The following theorem shows that, under $H_1$, $Y_n^\varepsilon$ asymptotically follows another power Poisson law.
\begin{Theorem}\label{power:case2}
If $\kappa\ge1$, $\varepsilon$ satisfies (\ref{vareps:condition}) and $(a-b)(a_\varepsilon-b_\varepsilon)<\frac{2(a+b)}{3}$,
then under $H_1$,
$Y_n^\varepsilon\overset{d}{\to}W_1^\varepsilon$ as $n\to\infty$,
where
\[
W_1^\varepsilon=\prod_{m=3}^\infty\left(1+\delta_m^\varepsilon\right)^{Z_m^1}\exp\left(-\lambda_m\delta_m^\varepsilon\right),\,\,\,\,
Z_m^1\overset{ind}{\sim}\textrm{Poisson}\left(\lambda_m(1+\delta_m)\right).
\]
Here, $\lambda_m$
and $\delta_m^\varepsilon$ are the same as in Theorem \ref{testing:consistency:case2}
and $\delta_m=\left(\frac{a-b}{a+b}\right)^m$.
\end{Theorem}
We notice that $W_1^\varepsilon$ differs from $W_0^\varepsilon$
only in the Poisson powers, i.e., $Z_m^1$ has larger means than $Z_m^0$.
Intuitively, the power of $Y_n^\varepsilon$ should increase when
such differences become substantial. 

Based on Theorems \ref{testing:consistency:case2} and \ref{power:case2},
we can derive the asymptotic power of $Y_n^\varepsilon$ as stated in the corollary below.
The power is an unexplicit function of $(a,b,\varepsilon)$.
  
\begin{Corollary}\label{cor:power:case2}
If $\kappa\ge1$,
$\varepsilon$ satisfies (\ref{vareps:condition}) and $(a-b)(a_\varepsilon-b_\varepsilon)<\frac{2(a+b)}{3}$,
then as $n\to\infty$, the power of $Y_n^\varepsilon$ satisfies
$P(\textrm{reject $H_0$}|\textrm{under $H_1$})\to P(a,b,\varepsilon)$, where $P(a,b,\varepsilon):=P(W_1^\varepsilon\ge w_\alpha^\varepsilon)$.
\end{Corollary}

\begin{Remark}

The value of $\varepsilon$ can be empirically selected.
Specifically, choose $\varepsilon$ to satisfy (\ref{vareps:condition}) and $(a-b)(a_\varepsilon-b_\varepsilon)<\frac{2(a+b)}{3}$
with $a,b$ therein replaced by their consistent estimators.
Existence of such consistent estimators is guaranteed by \cite{MNS15} when $\kappa>1$.
\end{Remark}

%\begin{Remark}
%The tuning parameter $\varepsilon$ can indeed reduce the variance of the LR test. In fact, 
%following the discussions in Remark \ref{rem:J95}, the second moments of $W_0^\varepsilon$
%and $W_1^\varepsilon$ are $\exp\left(\sum_{m=3}^\infty\lambda_m(\delta_m^\varepsilon)^2\right)$ and 
%%$\exp\left(\sum_{m=3}^\infty\lambda_m(1+\delta_m)(\delta_m^\varepsilon)^2\right)$, respectively.
%Both series are convergent thanks to (\ref{vareps:condition}). 
%\end{Remark}

\subsection{Power Analysis.}\label{sec:power:analysis}
Corollary \ref{cor:power:case2} derives an asymptotic power $P(a,b,\varepsilon)$ for $Y_n^\varepsilon$.
In this section, we further examine this power and demonstrate whether and when it can approach one.
It is challenging to directly analyze $P(a,b,\varepsilon)$ for fixed $a,b$ due to the lack of explicit expression. 
Instead, we will consider the relatively easier growing-degree regime ($a+b\to\infty$) and discuss its connection to existing work. 
Theorem \ref{lim:power:case:2} provides an explicit expression for the limit of $P(a,b,\varepsilon)$.
Let $\Phi(\cdot)$ denote the cumulative distribution function of standard normal variable
and $z_{1-\alpha}$ denote its $1-\alpha$ quantile, i.e.,
$\Phi(z_{1-\alpha})=1-\alpha$.

\begin{Theorem}\label{lim:power:case:2}
If $\kappa\ge1$ and $\varepsilon\in\left(0,\frac{a-b}{2}\right)$
satisfies, when $a+b\to\infty$, $\frac{(a_\varepsilon-b_\varepsilon)^2}{2(a+b)}\to k_1$ and
$\frac{(a-b)(a_\varepsilon-b_\varepsilon)}{2(a+b)}\to k_2$
for constants $k_1,k_2\in(0,1)$, then
$P(a,b,\varepsilon)\to\Phi\left(\frac{\sigma_2^2}{\sigma_1}-z_{1-\alpha}\right)$
as $a+b\to\infty$,
where $\sigma_l^2=\sum_{m=3}^\infty\frac{1}{2m}k_l^m=-\frac{1}{2}\left(\log(1-k_l)+k_l+\frac{1}{2}k_l^2\right)$, $l=1,2$. 
\end{Theorem}
We remark that the limit power $\Phi\left(\frac{\sigma_2^2}{\sigma_1}-z_{1-\alpha}\right)$ approaches one if $\kappa\to\infty$ (regardless of the choice of $\varepsilon$). 
To see this, note that
\begin{equation}\label{eq:power:anal}
\frac{\sigma_2^2}{\sigma_1}=-\frac{\log(1-k_2)+k_2+\frac{1}{2}k_2^2}{\sqrt{-2\left(\log(1-k_1)+k_1+\frac{1}{2}k_1^2\right)}}
\asymp\kappa^{3/2}.
\end{equation}
The above (\ref{eq:power:anal}) holds uniformly for $\varepsilon$ satisfying the conditions of Theorem \ref{lim:power:case:2}
and $\kappa^{3/2}$ on the right side is free of $\varepsilon$.
If $\kappa\to\infty$, then $\frac{\sigma_2^2}{\sigma_1}\to\infty$, and so $\Phi\left(\frac{\sigma_2^2}{\sigma_1}-z_{1-\alpha}\right)$ approaches one.
The power behavior merely relies on $\kappa$ while being free of $\varepsilon$.
Our result is closely relating to \cite{BM17} who investigate the asymptotic power of the classic LR test which nonetheless requires
$0<\kappa<1$. 
\cite{MS16} proposed an efficient method based on semidefinite program but their size and power are not
 explicitly quantifiable like ours.

\subsection{Monte-Carlo Approximation.}\label{sec:approx}
Despite its theoretically nice properties,
the test statistic $Y_n^\varepsilon$ might be computationally infeasible.
This can be easily seen from (\ref{lrt:case2}), i.e.,
$Y_n^\varepsilon$ can be viewed as the average of the quantity 
$g_n^\varepsilon(\sigma)$ over the entire space of configurations $\{\pm\}^n$,
where
\[
g_n^\varepsilon(\sigma)=\prod_{u<v}\left(\frac{p_{uv}^\varepsilon(\sigma)}{p_0}\right)^{A_{uv}}
\left(\frac{q_{uv}^\varepsilon(\sigma)}{q_0}\right)^{1-A_{uv}}.
\]
The computational effort for the direct averages
scales exponentially with $n$.
So an accurate and computationally efficient approximation of $Y_n^\varepsilon$ would be needed for practical use.

In this section, we consider the
classical Monte Carlo (MC) method which randomly chooses a set of $M$ configurations $\sigma[1],\ldots,\sigma[M]$ from $\{\pm\}^n$.
The MC method works directly under the bounded-degree regime. 
The average $Y_n^\varepsilon$ is naturally approximated by the sample mean of $g_n^\varepsilon(\sigma[l])$'s,
which may substantially reduce computational cost if $M\ll 2^n$.  
However, a small choice of $M$ may result in inaccurate approximation.
An interesting question is how small $M$ can be to ensure valid approximation.
More explicitly, we aim to find an order of $M$ such that the following approximation becomes valid:
%\begin{equation}\label{approx:lrt1}
%\widehat{Y}_n\equiv\frac{1}{M}\sum_{l=1}^M g_n(\sigma[l])=Y_n+o_P(1),
%\end{equation}
\begin{equation}\label{approx:lrt2}
\widehat{Y}_n^\varepsilon\equiv\frac{1}{M}\sum_{l=1}^M g_n^\varepsilon(\sigma[l])=Y_n^\varepsilon+o_P(1).
\end{equation}
The following theorem shows that the validity of (\ref{approx:lrt2}) is possible.
 
\begin{Theorem}\label{thm:approx2}
Suppose $\kappa\ge1$, $\varepsilon$ satisfies (\ref{vareps:condition}),
and $M\gg\exp\left(\frac{n\kappa_\varepsilon}{2}\right)$. 
Then (\ref{approx:lrt2}) holds under $H_0$.
Moreover, if $(a_\varepsilon-b_\varepsilon)(a-b)<a+b$, then (\ref{approx:lrt2}) holds under $H_1$. 
\end{Theorem}
According to Theorem \ref{thm:approx2}, (\ref{approx:lrt2}) becomes valid
when only $M\gg\exp\left(\frac{n\kappa_\varepsilon}{2}\right)$ configurations are used, e.g.,
$M\asymp(\log{n})^c\exp\left(\frac{n\kappa_\varepsilon}{2}\right)$ for a constant $c>0$.
When $\kappa_\varepsilon$ is close to zero, this may substantially reduce the computational cost
from $O(2^n)$ to nearly $O([\exp(\kappa_\varepsilon/2)]^n)$.
However, MC method still requires heavy computation.
A more efficient method would be highly useful. 

\section{Numerical Studies}

In this section, we examine the performance of
the proposed testing procedure through simulation studies in Section \ref{sec:sim},
and through real-world data sets in Section \ref{sec:realdata}.

\subsection{Simulation.}\label{sec:sim} 
The empirical performance of our test statistic $\widehat{Y}_n^\varepsilon$
is demonstrated through simulation studies.
We also compared our method with the spectral method proposed by Bickel and Sarkar \cite{BS16}
and the subgraph count method proposed by Gao and Lafferty \cite{GL17a}. Throughout we assume that both $a$ and $b$ are known. 
We evaluated the size and power of various methods at significance level 0.05.
For size, data were generated from $\mathcal{G}\left(n,\frac{a+b}{2n}\right)$.
For power, data were generated from $\mathcal{G}\left(n,\frac{a}{n},\frac{b}{n}\right)$.
Both size and power were calculated as proportions of rejections based on 500 independent experiments.

We examined various choices of $a,b,n,\varepsilon$ for $\widehat{Y}_n^\varepsilon$. 
For convenience, denote $a=2.5+c$, $b=2.5-c$. We chose $n=20,30,40,45$
and $(c,\varepsilon)=(2.10, 1.10), (2.15, 1.15), (2.25, 1.25), (2.35, 1.35)$.
The $\varepsilon$ in each case was chosen to be approximately $c/2$.
The corresponding values of SNR $\kappa=\frac{(a-b)^2}{2(a+b)}$ are $1.76, 1.85, 2.03, 2.21$.
We chose $100n^3e^{\frac{n\kappa_\varepsilon}{2}}$ samples for MC approximations according to Theorem \ref{thm:approx2} for 
calculating $\widehat{Y}_n^\varepsilon$.
Table \ref{taba2} summarizes the size and power of our test.
For all cases, the sizes of the $\widehat{Y}_n^\varepsilon$ are close to the 0.05 nominal level indicating the validity of the test.
For each choice of $(c,\varepsilon)$, the power increases along with $n$.
For any fixed $n$, the power increases as $\kappa$ increases, consistent
with Theorem \ref{lim:power:case:2} which states that the power should increase with $\kappa$.
  
\begin{table}[h]
\centering
		\begin{tabular}{ |c c |c |c|c|c| } 
			\hline
			$(c,\epsilon)$& $\kappa$	& $n=20$ 		&	 $n=30$ 	& $n=40$  & $n=45$\\
			\hline
			(2.10, 1.10) & 1.76 	&0.572 (0.058) &0.654 (0.042) 	&0.730 (0.050)& 0.812 (0.048)\\
			\hline
			(2.15, 1.15) & 1.85		&0.598 (0.054) 	&0.684 (0.052) &0.764 (0.050)  & 0.852 (0.040)\\
			\hline
			(2.25, 1.25) & 2.03		&0.626 (0.056) &0.704 (0.044)	&0.802 (0.042) & 0.910 (0.044)\\
			\hline
			(2.35, 1.35) & 2.21		&0.648 (0.044) 	&0.736 (0.042) &0.888 (0.058)  & 1.000 (0.042)\\
			\hline
		\end{tabular}
			\caption{Power (Size) of $\varepsilon$-LR test based on various choices of $c,\varepsilon,n$.}\label{taba2}
\end{table}

Tables \ref{tb::bs16} and \ref{tabgaoad} summarize the size and power of BS's spectral method and GL's subgraph count method.
It is worth mentioning that the sizes of both methods are free of $c$ since the null models under various values of $c$ are equivalent
and the sizes of both methods are uniquely determined by the common null model.
Due to the high sparsity of the simulated networks, subgraph counts are generally small,
and we obtained the critical values for GL's method based on resampling instead of using asymptotic distribution.
It is observed that both methods achieve smaller power than $\varepsilon$-LR while maintaining the correct size.

\begin{table}[h]
\centering
		\begin{tabular}{ |c c |c |c|c|c| } 
			\hline
			$c$& $\kappa$	& $n=20$ 		&	 $n=30$ 	& $n=40$  & $n=45$\\
			\hline  
			2.10 & 1.76 &0.308 (0.070)&0.260 (0.048)&0.266 (0.058)&0.300 (0.054)\\
			\hline
			2.15 & 1.85 &0.330 (0.070)&0.280 (0.048)&0.288 (0.058)&0.314 (0.054)\\
			\hline
			2.25 & 2.03 &0.364 (0.070)&0.336 (0.048)&0.348 (0.058)&0.362 (0.054)\\
			\hline
			2.35 & 2.21 &0.400 (0.070)&0.376 (0.048)&0.402 (0.058)&0.402 (0.054)\\
			\hline
		\end{tabular}
			\caption{Power (Size) of BS's spectral test based on various choices of $c,n$.}\label{tb::bs16}
\end{table}

\begin{table}[h]
\centering
		\begin{tabular}{ |c c |c|c|c|c|  } 
			\hline
			$c$	& $\kappa$ & $n=20$ 		&	 $n=30$ 	& $n=40$  & $n=45$\\
			\hline
			2.10	& 1.76 & 0.156 (0.05)  &0.216 (0.05) 	&0.196 (0.05)  &0.168 (0.05)  \\
			\hline
			2.15& 1.85 & 0.216 (0.05) &0.224 (0.05)	&0.204 (0.05) &0.206 (0.05) \\
			\hline
			2.25	& 2.03 & 0.296 (0.05) &0.234 (0.05)	&0.246 (0.05) &0.300 (0.05)\\
			\hline
			2.35	& 2.21 & 0.306 (0.05) &0.350 (0.05)	&0.328 (0.05) &0.336 (0.05)  \\
			\hline
		\end{tabular}
			\caption{Power (Size) of GL's subgraph count test based on various choices of $c,n$.}\label{tabgaoad}
\end{table}

\subsection{Real Data Analysis.}\label{sec:realdata}

In this section, we applied our procedure to analyze the political book data (\cite{Newman06})
which has 105 political books (nodes). Two books are connected if they were frequently co-purchased on Amazon.
This data was analyzed by \cite{ZLZ11} who detected three communities.
We used R package \textit{igraph} based on a spin-glass model and simulated annealing to recover their findings,
and denote the three communities by $C_I, C_{II}, C_{III}$
which contain 20, 44, 41 nodes, respectively. 
Books within the same community are expected to demonstrate similar political tendencies.
The aim of this study is to examine whether our method can detect the existence of the communities.
Our $\varepsilon$-LR method was based on $M=10^7e^{n\delta}$ MC samples with $\delta\approx 0.01$
and $\varepsilon\approx\frac{\hat{a}+\hat{b}}{2}$,
where $\widehat{a}=21.633$, $\widehat{b}=1.139$ are MLEs of $a,b$ under $H_1$.
We first examined whether $H_0$ is rejected (at 0.05) over each community.
Table \ref{tabrej:pbd:1} summarizes the results.
We find that all three methods rejected $H_0$ over $C_I$.
This incorrect decision might be due to the small size of the first community.
Moreover, $\varepsilon$-LR failed to reject $H_0$ over communities $C_{II}$, $C_{III}$;
BS rejected $H_0$ over $C_{II}$, $C_{III}$;
GL rejected $H_0$ over $C_{III}$ 
while failed to reject $H_0$ over $C_{II}$. 
\begin{table}[h]
\centering
		\begin{tabular}{ cccc  } 
Method & \multicolumn{3}{c}{P-value}\\ \hline
& $C_I$ & $C_{II}$ & $C_{III}$\\ \hline
$\varepsilon$-LR&0.000 &1.000 &1.000 \\
BS & 0.000 & 0.000 & 0.000\\
GL & 0.000 & 0.465 & 0.039\\ \hline
		\end{tabular}
			\caption{P-values of three methods over communities $C_I, C_{II}, C_{III}$
			for Political Book Data.}\label{tabrej:pbd:1}
\end{table}

We then examined whether $H_0$ is rejected for a subnetwork with nodes from two different communities.
In particular, we uniformly sampled 20 nodes out of $C_{II}$ without replacement, 
and combined it with $C_I$. 
Therefore, the combined  network has two communities of 20 nodes each.
We also examined the combination regimes $C_I$ \& $C_{II}$
and $C_{II}$ \& $C_{III}$ with 20 nodes uniformly sampled from $C_{II}$ in the former
and from both $C_{II}, C_{III}$ in the latter (so each combination has 40 nodes in total).
We repeated each combination regime 100 times and 
calculated the proportions that $H_0$ was rejected.
Results are summarized in Table \ref{tabrej:pbd:2}.
It is observed that all three methods rejected $H_0$ 100 times over $C_I$ \& $C_{II}$ hence
the success rates are 100\%.
The rejection rates of $\varepsilon$-LR over $C_I$ \& $C_{III}$ and over $C_{II}$ \& $C_{III}$ 
are $92\%$ and  $83\%$ respectively. The success rates of BS $100\%$,
and the success rates of GL are 98\% and 100\%, for both combination regimes.
\begin{table}[h]
\centering
		\begin{tabular}{ cccc  } 
Method & \multicolumn{3}{c}{Rejection Proportion}\\ \hline
& $C_I$ \& $C_{II}$ & $C_I$ \& $C_{III}$ & $C_{II}$ \& $C_{III}$\\ \hline
$\varepsilon$-LR&100\% &92\% &83\% \\
BS & 100\% & 100\% & 100\%\\
GL & 100\% & 98\% & 100\%\\ \hline
		\end{tabular}
			\caption{Rejection proportions by three methods in different combination regimes
			for Political Book Data.}\label{tabrej:pbd:2}
\end{table}

\section{Discussions}

The work of \cite{DKMZ11} implies that
extension of the current work to multi-community setting is highly important but nontrivial. 
As far as we know, only a few works address such settings but mostly in community detection.
For instance, \cite{NN14} provides a sufficient
condition for impossible detection;
\cite{AS17} presents an information-theoretic phase transition for the SNR to yield successful detection
which strengthens the work of \cite{NN14}. 

The test statistic $\widehat{Y}_n^\varepsilon$ can be viewed as a type of partition function over Gibbs field.
Popular approximations of partition functions in statistical physics include MC approximations and mean-field approximations.
This paper only considers the former while leaves the latter as a future topic. 
Mean-field approximation has proven to work well in dense magnetism
such as Curie-Weiss model (see \cite{W1907}).
Recently, validity of mean-field approximation was established by \cite{BM17-2}
in the sparser settings which satisfy the so-called ``mean-field assumption,''
i.e., the trace of the squared adjacency matrix is $o_P(n)$.
This assumption fails in our setting in that the trace becomes $O_P(n)$.
Additional theory is needed to extend the results of \cite{BM17-2}.

\section{Appendix: Proofs}
In this section, we prove the main results of this paper.
Our asymptotic results are derived based on the following Proposition \ref{basic:prop} 
which was proved by Janson in \cite{J95}.
For arbitrary non-negative integer $x$, let
$[x]_j$ denote the descending factorial $x(x-1)\cdots(x-j+1)$.

\begin{Proposition}\label{basic:prop}
Let $\lambda_i>0$, $i=1,2,\ldots$, be constants and suppose that for each $n$ there are random variables
$X_{in}$, $i=1,2,\ldots$, and $Y_n$ (defined on the same probability space) such that $X_{in}$ is non-negative integer valued
and $E\{Y_n\}\neq0$ (at least for large $n$), and furthermore the following conditions are satisfied:
\begin{itemize}
\item[(A1)] $X_{in}\overset{d}{\to}Z_i$ as $n\to\infty$, jointly for all $i$, where $Z_i\sim\textrm{Poisson}(\lambda_i)$
are independent Poisson random variables;
\item[(A2)] $\E\{Y_n[X_{1n}]_{j_1}\cdots[X_{kn}]_{j_k}\}/\E\{Y_n\}\to\prod_{i=1}^k\mu_i^{j_i}$, as $n\to\infty$,
for some $\mu_i\ge0$ and every finite sequence $j_1,\ldots,j_k$ of non-negative integers;
\item[(A3)] $\sum_{i=1}^\infty\lambda_i\delta_i^2<\infty$, where $\delta_i=\mu_i/\lambda_i-1$;
\item[(A4)] $\E\{Y_n^2\}/(\E\{Y_n\})^2\to\exp\left(\sum_{i=1}^\infty\lambda_i\delta_i^2\right)$.
\end{itemize}
Then 
\[
\frac{Y_n}{\E\{Y_n\}}\overset{d}{\to}W\equiv\prod_{i=1}^\infty(1+\delta_i)^{Z_i}\exp(-\lambda_i\delta_i),\,\,\,\,\textrm{as $n\to\infty$}.
\]
\end{Proposition}

\begin{Remark}\label{rem:J95}
Janson (1995) \cite{J95} showed that the infinite product defining $W$ in Proposition \ref{basic:prop}
converges in $L^2$ a.s. with $\E W=1$ and $\E W^2=\exp\left(\sum_{i\ge1}\lambda_i\delta_i^2\right)$. 
\end{Remark}

%\subsection{Proofs in Section \ref{sec:case1}}
Before proofs, we need the following lemma.
\begin{Lemma}\label{lemma:short:cycles}
For a random graph $G$ with vertex $1,\ldots,n$,
let $X_{mn}$ be the number of $m$-cycles of $G$, for $m\ge3$.
Let $\lambda_m=\frac{1}{2m}\left(\frac{a+b}{2}\right)^m$
and $\delta_m=\left(\frac{a-b}{a+b}\right)^m$.
\begin{enumerate}
\item\label{lemma:short:cycles:1} Under $G\sim\mathcal{G}(n,p_0)$, for any $k\ge3$,
$\{X_{mn}\}_{m=3}^k$ jointly converge 
to independent Poisson variables with mean $\lambda_m$.
\item\label{lemma:short:cycles:2} Under $G\sim\mathcal{G}\left(n,\frac{a}{n},\frac{b}{n}\right)$, for any $k\ge3$,
$\{X_{mn}\}_{m=3}^k$ jointly converge 
to independent Poisson variables with mean $\lambda_m(1+\delta_m)$.
\end{enumerate}
\end{Lemma}
\begin{proof}[Proof of Lemma \ref{lemma:short:cycles}]
The first part was well known (see \cite{MNS15}). We only prove the second part.
 
Denote $\E_1$ the expectation based on hypothesis $H_1$. Let $H$ be a graph on a subset of $[n]$ with vertex set $\mathcal{V}(H)$
and edge set $\mathcal{E}(H)$. Use $1_H$ to denote the 0-1 random variable
that is 1 when $\mathcal{E}(H)\subseteq\mathcal{E}(G)$ and $P(H)$ for the probability that
$1_H=1$. 
For $3\le m\le k$,
let $H_{m1},\ldots,H_{mj_m}$ be a $j_m$-tuple of distinct $m$-cycles. 
Then
\[
\prod_{m=3}^k[X_{mn}]_{j_m}=\sum_{(H_{mi})}\prod_{m=3}^k\prod_{i=1}^{j_m}1_{H_{mi}},
\]
where the sum ranges over all tuples of distinct cycles $\{H_{mi}: 3\le m\le k, 1\le i\le j_m\}$;
each $H_{mi}$ is an $m$-cycle and all cycles are distinct.
Let $A$ be the set of all such tuples of cycles for which the cycles are vertex-disjoint
and let $\bar{A}$ be its complement, i.e., any tuple of $\bar{A}$ contains two cycles with at least one common
vertex. Then 
\begin{eqnarray}\label{proof:lemma:short:cycles:eqn1}
\E_1\prod_{m=3}^k[X_{mn}]_{j_m}&=&\sum_{(H_{mi})}\E_1\prod_{m=3}^k\prod_{i=1}^{j_m}1_{H_{mi}}\nonumber\\
&=&\sum_{(H_{mi})\in A}\E_1\prod_{m=3}^k\prod_{i=1}^{j_m}1_{H_{mi}}+\sum_{(H_{mi})\in \bar{A}}\E_1\prod_{m=3}^k\prod_{i=1}^{j_m}1_{H_{mi}}
\end{eqnarray}
Since the number of $m$-cycles on a graph of $s$ vertexes is 
$\frac{s!}{(s-m)!2m}$ (two directions and $m$ distinct starting vertexes
give us $2m$ the same $m$-cycles), one gets that
$|A|=\frac{n!}{(n-M)!}\prod_{m=3}^k\left(\frac{1}{2m}\right)^{j_m}$
with $M=\sum_{m=3}^k mj_m$ (see also \cite[Chapter 4]{B01} for more complete derivation).
Meanwhile, take $\tau$ uniformly from $\{\pm\}^n$
and define $\tau^{mi}$ be the restriction of $\tau$ on the vertexes of $H_{mi}$,
and define $N_{mi}=\sum_{(u,v)\in\mathcal{E}(H_{mi})}1(\tau^{mi}_u\neq\tau^{mi}_v)$.
The $\tau^{mi}$'s are independent thanks to the vertex disjointness of $H_{mi}$'s.
Following \cite[Lemma 3.3]{MNS15} one can show that
$P(N_{mi}=l)=2^{-m+1}{m\choose l}$ for even $l\in[0,m]$ and zero for odd $l$.
Then one has
\begin{eqnarray*}
\E_1\prod_{m=3}^k\prod_{i=1}^{j_m}1_{H_{mi}}&=&\E_\tau\E_1\{\prod_{m=3}^k\prod_{i=1}^{j_m}1_{H_{mi}}|\tau\}\\
&=&\E_\tau\prod_{m=3}^k\prod_{i=1}^{j_m}\prod_{(u,v)\in\mathcal{E}(H_{mi})}\left(\frac{a}{n}\right)^{1(\tau_u=\tau_v)}\left(\frac{b}{n}\right)^{1(\tau_u\neq\tau_v)}\\
&=&\E_\tau\prod_{m=3}^k\prod_{i=1}^{j_m}\prod_{(u,v)\in\mathcal{E}(H_{mi})}\left(\frac{a}{n}\right)^{1(\tau^{mi}_u=\tau^{mi}_v)}\left(\frac{b}{n}\right)^{1(\tau^{mi}_u\neq\tau^{mi}_v)}\\
&=&\E_\tau\prod_{m=3}^k\prod_{i=1}^{j_m}\left(\frac{a}{n}\right)^{m-N_{mi}}\left(\frac{b}{n}\right)^{N_{mi}}.
\end{eqnarray*}
Since $\tau$ is broken into disjoint and independent $(\tau^{mi})_{3\le m\le k, 1\le i\le j_m}$, the above is equal to
\begin{eqnarray*}
&&\prod_{m=3}^k\prod_{i=1}^{j_m}\E_{\tau^{mi}}\left(\frac{a}{n}\right)^{m-N_{mi}}\left(\frac{b}{n}\right)^{N_{mi}}\\
&=&\prod_{m=3}^k\prod_{i=1}^{j_m}2^{-m}\left[\left(\frac{a+b}{n}\right)^m+\left(\frac{a-b}{n}\right)^m\right]
=n^{-M}\prod_{m=3}^k\left[\left(\frac{a+b}{2}\right)^m\left(1+\delta_m\right)\right].
\end{eqnarray*}
Then the first part of (\ref{proof:lemma:short:cycles:eqn1})
becomes 
\begin{eqnarray*}
&&|A|\times n^{-M}\prod_{m=3}^k\left[\left(\frac{a+b}{2}\right)^m\left(1+\delta_m\right)\right]
\\
&=&\frac{n!}{(n-M)!n^M}\prod_{m=3}^k
(\lambda_m(1+\delta_m))^{j_m}\overset{n\to\infty}{\rightarrow}\prod_{m=3}^k
(\lambda_m(1+\delta_m))^{j_m}.
\end{eqnarray*}

On the other hand, for any $(H_{mi})\in \bar{A}$, 
$H:=\cup H_{mi}$ has at most $M-1$ vertexes and $M$ edges,
and $|\mathcal{E}(H)|>|\mathcal{V}(H)|$.
Since 
\[
\E_1\{\prod_{m=3}^k\prod_{i=1}^{j_m}1_{H_{mi}}|\tau\}=
\prod_{(u,v)\in\mathcal{E}(H)}\left(\frac{a}{n}\right)^{1(\tau_u=\tau_v)}\left(\frac{b}{n}\right)^{1(\tau_u\neq\tau_v)}\le\left(\frac{\max\{a,b\}}{n}\right)^{|\mathcal{E}(H)|},
\]
and there are ${n\choose |\mathcal{V}(H)|}|\mathcal{V}(H)|!$ graphs isomorphic to $H$,
then 
\[
\sum_{\textrm{$H'$ is isomorphic to $H$}}\E_1\{1_{H'}|\tau\}\le
\left(\frac{\max\{a,b\}}{n}\right)^{|\mathcal{E}(H)|}{n\choose |\mathcal{V}(H)|}|\mathcal{V}(H)|!\to0.
\]
Since there are a bounded number of isomorphism classes, the second part of (\ref{proof:lemma:short:cycles:eqn1})
tends to zero as $n\to\infty$. Hence, $\E_1\prod_{m=3}^k[X_{mn}]_{j_m}\to\prod_{m=3}^k
(\lambda_m(1+\delta_m))^{j_m}.$ for any $k\ge3$ and integers $j_3,\ldots,j_k$.
It follows by \cite[Lemma 2.8]{W99} that the desirable result holds.
\end{proof}

\subsection{Proofs in Section \ref{sec:ppl}}
\begin{proof}[Proof of Theorem \ref{testing:consistency:case2}]
Let $\E_0$ denote the expectations under hypotheses $H_0$.
We will use Proposition \ref{basic:prop} to prove the result,
for which we will check the Conditions A1 to A4 therein. 
Some of the details are rooted in \cite{MNS15}. To ease reading, we provide the detailed proofs.
Obviously, $\E_0Y_n^\varepsilon=1$.

Let $X_{mn}$ be the number of $m$-cycles of $G\sim\mathcal{G}(n,p_0)$, for $m\ge3$.
Following Lemma \ref{lemma:short:cycles} Part \ref{lemma:short:cycles:1},
for any $k\ge3$, $\{X_{mn}\}_{m=3}^k$ jointly converge 
to independent Poisson variables with mean $\lambda_m=\frac{1}{2m}\left(\frac{a+b}{2}\right)^m$.
This verifies Condition A1.

To check Condition A2, let $H=(H_{mi})_{3\le m\le k, 1\le i\le j_m}$ be a tuple of short cycles of disjoint vertexes;
each $H_{mi}$ is an $m$-cycle, $M=\sum_{m=3}^k mj_m$, and the vertexes of $H_{mi}$'s are disjoint.
Let $\sigma^{1mi},\sigma^{2mi}$ be the restrictions of
$\sigma$ over $\mathcal{V}(H_{mi})$ and $[n]\backslash\mathcal{V}(H_{mi})$,
and $\sigma^1,\sigma^2$ be the restrictions of $\sigma$ over $\mathcal{V}(H)$ and $[n]\backslash\mathcal{V}(H)$.
By direct examinations we have
\begin{eqnarray}
\E_0 Y_n^\varepsilon 1_H
&=&2^{-n}\sum_{\sigma\in\{\pm\}^n}\E_0 1_H
\prod_{u<v}\left(\frac{p_{uv}^\varepsilon(\sigma)}{p_0}\right)^{A_{uv}}\left(\frac{q_{uv}^\varepsilon(\sigma)}{q_0}\right)^{1-A_{uv}}
\nonumber\\
&=&2^{-n}\sum_{\sigma\in\{\pm\}^n}\E_0 1_H\prod_{(u,v)\in\mathcal{E}(H)}\left(\frac{p_{uv}^\varepsilon(\sigma)}{p_0}\right)^{A_{uv}}\left(\frac{q_{uv}^\varepsilon(\sigma)}{q_0}\right)^{1-A_{uv}}\nonumber\\
&&\times\prod_{(u,v)\in\overline{\mathcal{E}(H)}}\left(\frac{p_{uv}^\varepsilon(\sigma)}{p_0}\right)^{A_{uv}}
\left(\frac{q_{uv}^\varepsilon(\sigma)}{q_0}\right)^{1-A_{uv}}\nonumber\\
&=&2^{-n}\sum_{\sigma\in\{\pm\}^n}\E_0 1_H\prod_{(u,v)\in\mathcal{E}(H)}\left(\frac{p_{uv}^\varepsilon(\sigma)}{p_0}\right)^{A_{uv}}\left(\frac{q_{uv}^\varepsilon(\sigma)}{q_0}\right)^{1-A_{uv}}.\label{testing:consistency:eqn1}
\end{eqnarray}
Since $\sigma$ is broken into $\sigma^1$ and $\sigma^2$ which are supported on $\mathcal{V}$ and its complement respectively,
and $p_{uv}^\varepsilon(\sigma)$, $q_{uv}^\varepsilon(\sigma)$ only depend on $\sigma^1$ when $(u,v)\in\mathcal{E}(H)$,
(\ref{testing:consistency:eqn1}) is equal to the following
\begin{eqnarray}
&&2^{-n}\sum_{\sigma^1\in\{\pm\}^{\mathcal{V}(H)}}\sum_{\sigma^2\in\{\pm\}^{[n]\backslash\mathcal{V}(H)}}
\E_0 1_H\prod_{(u,v)\in\mathcal{E}(H)}\left(\frac{p_{uv}^\varepsilon(\sigma^1)}{p_0}\right)^{A_{uv}}
\left(\frac{q_{uv}^\varepsilon(\sigma^1)}{q_0}\right)^{1-A_{uv}}\nonumber\\
&=&2^{-M}\sum_{\sigma^1\in\{\pm\}^{\mathcal{V}(H)}}\E_0 1_H\prod_{(u,v)\in\mathcal{E}(H)}
\left(\frac{p_{uv}^\varepsilon(\sigma^1)}{p_0}\right)^{A_{uv}}
\left(\frac{q_{uv}^\varepsilon(\sigma^1)}{q_0}\right)^{1-A_{uv}}.\label{testing:consistency:eqn1-1}
\end{eqnarray}
Since $1_H=1$ implies $\mathcal{E}(H)\subset\mathcal{E}(G)$, any $(u,v)\in\mathcal{V}(H)$
leads to $A_{uv}=1$. Meanwhile, $\E_0 1_H=p_0^M$, hence (\ref{testing:consistency:eqn1-1}) equals
\begin{eqnarray*}
&&2^{-M}p_0^M\sum_{\sigma^1\in\{\pm\}^{\mathcal{V}(H)}}\prod_{(u,v)\in\mathcal{E}(H)}\left(\frac{p_{uv}^\varepsilon(\sigma^1)}{p_0}\right)\nonumber\\
&=&\E_{\sigma^1}\prod_{(u,v)\in\mathcal{E}(H)}p_{uv}^\varepsilon(\sigma^1)=\prod_{m=3}^k\prod_{i=1}^{j_m}\E_{\sigma^{1mi}}\prod_{(u,v)\in\mathcal{E}(H_{mi})}
p_{uv}^\varepsilon(\sigma^{1mi})
=n^{-M}\prod_{m=3}^k\prod_{i=1}^{j_m}\E_{\sigma^{1mi}}a_\varepsilon^{m-N_{mi}}b_\varepsilon^{N_{mi}},\nonumber
\end{eqnarray*} 
where $N_{mi}=\sum_{(u,v)\in\mathcal{E}(H_{mi})}1(\sigma_u^{1mi}\neq\sigma^{1mi}_v)$,
the number of edges over $H_{mi}$ with distinct end points.
Following the proof of \cite[Lemma 3.3]{MNS15},
\begin{eqnarray*}
\E_{\sigma^{1mi}}a_\varepsilon^{m-N_{mi}}b_\varepsilon^{N_{mi}}=2^{-m}\left[(a_\varepsilon+b_\varepsilon)^m+(a_\varepsilon-b_\varepsilon)^m\right]=
\left(\frac{a+b}{2}\right)^m\left(1+\left(\frac{a_\varepsilon-b_\varepsilon}{a+b}\right)^m\right).
\end{eqnarray*}
Hence,
\[
\E_0 Y_n^\varepsilon 1_H=n^{-M}\prod_{m=3}^k\left(
\left(\frac{a+b}{2}\right)^m\left(1+\left(\frac{a_\varepsilon-b_\varepsilon}{a+b}\right)^m\right)
\right)^{j_m}.
\]
Let $A$ be the set of tuples $(H_{mi})_{3\le m\le k, 1\le i\le j_m}$ for which the cycles are vertex-disjoint
and let $\bar{A}$ be its complement.
Using $|A|=\frac{n!}{(n-M)!}\prod_{m=3}^k\left(\frac{1}{2m}\right)^{j_m}$ (see proof of Lemma \ref{lemma:short:cycles}) we get that
\begin{eqnarray*}
\sum_{H\in A}\E_0 Y_n^\varepsilon 1_H&=&|A|n^{-M}\prod_{m=3}^k\left(
\left(\frac{a+b}{2}\right)^m\left(1+\left(\frac{a_\varepsilon-b_\varepsilon}{a+b}\right)^m\right)
\right)^{j_m}\\
&\overset{n\to\infty}{\rightarrow}&\prod_{m=3}^k
\left(\frac{1}{2m}\left(\frac{a+b}{2}\right)^m\left(1+\left(\frac{a_\varepsilon-b_\varepsilon}{a+b}\right)^m\right)\right)^{j_m}\\
&=&\prod_{m=3}^k\left(\lambda_m(1+\delta_m^\varepsilon)\right)^{j_m},
\end{eqnarray*}
where $\delta_m^\varepsilon=\left(\frac{a_\varepsilon-b_\varepsilon}{a+b}\right)^m$.
Similar to (\ref{testing:consistency:eqn1}) one gets that, for $H\in\bar{A}$,
\begin{eqnarray*}
\E_0Y_n^\varepsilon 1_H&=&2^{-n}\sum_{\sigma\in\{\pm\}^n}\E_0 1_H\prod_{(u,v)\in\mathcal{E}(H)}\left(\frac{p_{uv}^\varepsilon(\sigma)}{p_0}\right)^{A_{uv}}\left(\frac{q_{uv}^\varepsilon(\sigma)}{q_0}\right)^{1-A_{uv}}\\
&=&2^{-n}\sum_{\sigma\in\{\pm\}^n}\prod_{(u,v)\in\mathcal{E}(H)}\frac{p_{uv}^\varepsilon(\sigma)}{p_0}\times P_0(H)\\
&\le&a_\varepsilon^{|\mathcal{E}(H)|}n^{-|\mathcal{E}(H)|},
\end{eqnarray*}
where the last inequality follows from $p_{uv}^\varepsilon(\sigma)\le a_\varepsilon/n$ and $P_0(H)=p_0^{|\mathcal{E}(H)|}$.
So 
\[
\sum_{\textrm{$H'$ is isomorphic to $H$}}\E_0Y_n^\varepsilon 1_H
\le a_\varepsilon^{|\mathcal{E}(H)|}n^{-|\mathcal{E}(H)|}{n\choose |\mathcal{V}(H)|}|\mathcal{V}(H)|!\to0,
\]
which leads to $\sum_{H\in\bar{A}}\E_0Y_n^\varepsilon 1_H\to0$ using a similar argument as the proof of Lemma \ref{lemma:short:cycles} Part \ref{lemma:short:cycles:2}.
So as $n\to\infty$,
\[
\E_0 Y_n^\varepsilon [X_{3n}]_{j_3}\cdots[X_{kn}]_{j_k}=\sum_{H\in A}\E_0Y_n^\varepsilon 1_H+\sum_{H\in\bar{A}}\E_0Y_n^\varepsilon 1_H\to
\prod_{m=3}^k\left(\lambda_m(1+\delta_m^\varepsilon)\right)^{j_m},
\]
which verifies Condition A2.

Condition A3 holds due to the following trivial fact:
\[
\sum_{m\ge3}\lambda_m(\delta_m^\varepsilon)^2=\sum_{m\ge3}\frac{1}{2m}\left(\frac{(a_\varepsilon-b_\varepsilon)^2}{2(a+b)}\right)^m=
\sum_{m\ge3}\frac{1}{2m}\left(\frac{(a-b-2\varepsilon)^2}{2(a+b)}\right)^m
<\infty.
\]

In the end let us check Conditions A4. 
Let $N_{uv}^{\sigma\tau}=1(\sigma_u=\sigma_v)+1(\tau_u=\tau_v)$.
Note that
\begin{eqnarray*}
\E_0 (Y_n^\varepsilon)^2&=&
4^{-n}\sum_{\sigma,\tau\in\{\pm\}^n}\prod_{u<v}\E_0
\left(\frac{p_{uv}^\varepsilon(\sigma)p_{uv}^\varepsilon(\tau)}{p_0^2}\right)^{A_{uv}}
\left(\frac{q_{uv}^\varepsilon(\sigma)q_{uv}^\varepsilon(\tau)}{q_0^2}\right)^{1-A_{uv}}\\
&=&4^{-n}\sum_{\sigma,\tau\in\{\pm\}^n}\prod_{u<v}\left(
\frac{p_{uv}^\varepsilon(\sigma)p_{uv}^\varepsilon(\tau)}{p_0}+
\frac{q_{uv}^\varepsilon(\sigma)q_{uv}^\varepsilon(\tau)}{q_0}
\right)
\\
&=&4^{-n}\sum_{\sigma,\tau\in\{\pm\}^n}\prod_{u<v}
\left(\frac{1}{p_0}\left(\frac{a_\varepsilon}{n}\right)^{N_{uv}^{\sigma\tau}}
\left(\frac{b_\varepsilon}{n}\right)^{2-N_{uv}^{\sigma\tau}}+\frac{1}{q_0}
\left(1-\frac{a_\varepsilon}{n}\right)^{N_{uv}^{\sigma\tau}}
\left(1-\frac{b_\varepsilon}{n}\right)^{2-N_{uv}^{\sigma\tau}}\right)\\
&=&4^{-n}\sum_{\sigma,\tau\in\{\pm\}^n}
\prod_{N_{uv}^{\sigma\tau}=0}\left(\frac{1}{p_0}\left(\frac{b_\varepsilon}{n}\right)^2
+\frac{1}{q_0}\left(1-\frac{b_\varepsilon}{n}\right)^2\right)\\
&&\times\prod_{N_{uv}^{\sigma\tau}=2}\left(\frac{1}{p_0}\left(\frac{a_\varepsilon}{n}\right)^2
+\frac{1}{q_0}\left(1-\frac{a_\varepsilon}{n}\right)^2\right)\\
&&\times\prod_{N_{uv}^{\sigma\tau}=1}\left(\frac{1}{p_0}\left(\frac{a_\varepsilon}{n}\right)
\left(\frac{b_\varepsilon}{n}\right)+\frac{1}{q_0}
\left(1-\frac{a_\varepsilon}{n}\right)
\left(1-\frac{b_\varepsilon}{n}\right)\right).
\end{eqnarray*}
It is easy to check that
\begin{eqnarray}\label{expression:gamman:varepsilon}
\frac{1}{p_0}\left(\frac{b_\varepsilon}{n}\right)^2
+\frac{1}{q_0}\left(1-\frac{b_\varepsilon}{n}\right)^2&=&1+\gamma_n^\varepsilon+O(n^{-3})\nonumber\\
\frac{1}{p_0}\left(\frac{a_\varepsilon}{n}\right)^2
+\frac{1}{q_0}\left(1-\frac{a_\varepsilon}{n}\right)^2&=&1+\gamma_n^\varepsilon+O(n^{-3})\nonumber\\
\frac{1}{p_0}\left(\frac{a_\varepsilon}{n}\right)
\left(\frac{b_\varepsilon}{n}\right)+\frac{1}{q_0}
\left(1-\frac{a_\varepsilon}{n}\right)
\left(1-\frac{b_\varepsilon}{n}\right)&=&1-\gamma_n^\varepsilon+O(n^{-3}),
\end{eqnarray}
where
$\gamma_n^\varepsilon=\frac{\kappa_\varepsilon}{n}+\frac{(a_\varepsilon-b_\varepsilon)^2}{4n^2}$,
$\kappa_\varepsilon=\frac{(a_\varepsilon-b_\varepsilon)^2}{2(a+b)}$.
Let
\[
s_+=\#\{(u,v): u<v, \sigma_u\sigma_v\tau_u\tau_v=+\},
s_-=\#\{(u,v): u<v, \sigma_u\sigma_v\tau_u\tau_v=-\}.
\]
Let $\rho=\frac{1}{n}\sum_{u=1}^n\sigma_u\tau_u$.
Following \cite{MNS15},
we have
$s_+=\frac{n^2}{4}(1+\rho^2)-\frac{n}{2}$
and $s_-=\frac{n^2}{4}(1-\rho^2)$.
Then using the approximation technique in \cite{MNS15}, i.e., Lemmas 5.3, 5.4, 5.5 therein, it holds that
\begin{eqnarray*}
\E_0(Y_n^\varepsilon)^2&=&4^{-n}\sum_{\sigma,\tau}(1+\gamma_n^\varepsilon+O(n^{-3}))^{s_+}(1-\gamma_n^\varepsilon+O(n^{-3}))^{s_-}\\
&=&(1+o(1))4^{-n}\sum_{\sigma,\tau}(1+\gamma_n^\varepsilon)^{\frac{n^2}{4}(1+\rho^2)-\frac{n}{2}}(1-\gamma_n^\varepsilon)^{\frac{n^2}{4}(1-\gamma^2)}\\
&=&(1+o(1))\exp\left(-\kappa_\varepsilon^2/4-\kappa_\varepsilon/2\right)4^{-n}\sum_{\sigma,\tau}\exp\left(
\frac{\rho^2}{2}\left(n\kappa_\varepsilon+\frac{(a_\varepsilon-b_\varepsilon)^2}{4}\right)
\right)\\
&=&(1+o(1))\exp\left(-\kappa_\varepsilon^2/4-\kappa_\varepsilon/2\right)\E_{\sigma\tau}
\exp\left(\frac{\rho^2}{2}\left(n\kappa_\varepsilon+\frac{(a_\varepsilon-b_\varepsilon)^2}{4}\right)\right)\\
&\overset{n\to\infty}{\rightarrow}&\exp\left(-\kappa_\varepsilon^2/4-\kappa_\varepsilon/2\right)(1-\kappa_\varepsilon)^{-1/2}=
\exp\left(\sum_{m=3}^\infty\lambda_m(\delta_m^\varepsilon)^2\right).
\end{eqnarray*}
This verifies Condition A4. The result of Theorem \ref{testing:consistency:case2} follows from Proposition \ref{basic:prop}.
\end{proof}

\begin{proof}[Proof of Theorem \ref{power:case2}]

Let $X_{mn}$ be the number of $m$-cycles of $G$, for $m\ge3$.
Let $\lambda_m=\frac{1}{2m}\left(\frac{a+b}{2}\right)^m$
and $\delta_m=\left(\frac{a-b}{a+b}\right)^m$.
It follows by Lemma \ref{lemma:short:cycles} Part \ref{lemma:short:cycles:2} that,
under $H_1$, $\{X_{mn}\}_{m=3}^k$ jointly converge 
to independent Poisson variables with mean $\lambda_m(1+\delta_m)$, verifying Condition A1 of Proposition \ref{basic:prop}.
This leaves us to check Conditions A2 to A4.
Let $M=\sum_{m=3}^k mj_m$ for integers $j_3,\ldots,j_k$ and $k\ge3$.

\textbf{Check Condition A2}. 
Denote $\E_1$ the expectation based on hypothesis $H_1$.
Let $X_{mn}$ be the number of $m$-cycles of $G$, for $m\ge3$
and $[x]_j$ be the descending factorial. Define $M=\sum_{m=3}^k mj_m$
for $k\ge 3$ and integers $j_3,\ldots,j_k$.
To check A2, notice that
\begin{eqnarray}
\E_1Y_n^\varepsilon [X_{3n}]_{j_3}\cdots[X_{kn}]_{j_k}&=&\sum_{(H_{mi})_{3\le m\le k, 1\le i\le j_m}}
\E_1Y_n^\varepsilon 1_{\cup H_{mi}}\label{thm:power:case2:eqn0}\\
&=&\sum_{(H_{mi})\in A}
\E_1Y_n^\varepsilon 1_{\cup H_{mi}}+\sum_{(H_{mi})\in\bar{A}}
\E_1Y_n^\varepsilon 1_{\cup H_{mi}},\label{thm:power:case2:eqn-1}
\end{eqnarray}
where the sum in (\ref{thm:power:case2:eqn0}) ranges over $\mathcal{H}$, the collection of
all $M$-tuples of cycles
$(H_{mi})_{3\le m\le k, 1\le i\le j_m}$ with each $H_{mi}$ an $m$-cycle,
and $A$ in the sum of (\ref{thm:power:case2:eqn-1}) is the set of such tuples for which the cycles are vertex-disjoint
and let $\bar{A}=\mathcal{H}\backslash A$, i.e., $\bar{A}$ contains $M$-tuples of cycles $(H_{mi})_{3\le m\le k, 1\le i\le j_m}$
with at least one common vertex among those cycles.
Let us look at the first part of (\ref{thm:power:case2:eqn-1}).
Take $\tau$ uniformly distributed from $\{\pm\}^n$.
For any $H=(H_{mi})\in\mathcal{H}$, define $\tau^1,\tau^2$ to be the
restrictions of $\tau$ over $\mathcal{V}(H)$ and $[n]\backslash\mathcal{V}(H)$ respectively.

One can check that, for any $H\in\mathcal{H}$,
\begin{eqnarray}
&&\E_1\{Y_n^\varepsilon 1_H|\tau\}\nonumber\\
&=&\E_1\left\{1_H 2^{-n}\sum_{\sigma}\prod_{u<v}\left(\frac{p_{uv}^\varepsilon(\sigma)}{p_0}\right)^{A_{uv}}
\left(\frac{q_{uv}^\varepsilon(\sigma)}{q_0}\right)^{1-A_{uv}}\bigg|\tau\right\}\nonumber\\
&=&2^{-n}\sum_{\sigma}\E_1\left\{1_H\prod_{u<v}\left(\frac{p_{uv}^\varepsilon(\sigma)}{p_0}\right)^{A_{uv}}\left(\frac{q_{uv}^\varepsilon(\sigma)}{q_0}\right)^{1-A_{uv}}\bigg|\tau\right\}\nonumber\\
&=&2^{-n}\sum_{\sigma}\E_1\left\{1_H\prod_{(u,v)\in\mathcal{E}(H)}\left(\frac{p_{uv}^\varepsilon(\sigma)}{p_0}\right)^{A_{uv}}\left(\frac{q_{uv}^\varepsilon(\sigma)}{q_0}\right)^{1-A_{uv}}
\bigg|\tau\right\}\nonumber\\
&&\times\prod_{(u,v)\in\overline{\mathcal{E}(H)}}\E_1\left\{\left(\frac{p_{uv}^\varepsilon(\sigma)}{p_0}\right)^{A_{uv}}\left(\frac{q_{uv}^\varepsilon(\sigma)}{q_0}\right)^{1-A_{uv}}\bigg|\tau\right\}\nonumber\\
&=&2^{-n}\sum_{\sigma}\E_1\{1_H|\tau\}\prod_{(u,v)\in\mathcal{E}(H)}\left(\frac{p_{uv}^\varepsilon(\sigma)}{p_0}\right)\prod_{(u,v)\in\overline{\mathcal{E}(H)}}
\left(\frac{p_{uv}^\varepsilon(\sigma)p_{uv}(\tau)}{p_0}+\frac{q_{uv}^\varepsilon(\sigma)q_{uv}(\tau)}{q_0}\right)\nonumber\\
&=&2^{-n}\sum_{\sigma}\prod_{(u,v)\in\mathcal{E}(H)}\left(\frac{p_{uv}^\varepsilon(\sigma^1)p_{uv}(\tau^1)}{p_0}\right)
\prod_{(u,v)\in\overline{\mathcal{E}(H)}}
\left(\frac{p_{uv}^\varepsilon(\sigma)p_{uv}(\tau)}{p_0}+\frac{q_{uv}^\varepsilon(\sigma)q_{uv}(\tau)}{q_0}\right),\label{thm:power:case1:eqn24}
\end{eqnarray}
which leads to that
\begin{eqnarray}
&&\E_1 Y_n^\varepsilon 1_H\nonumber\\
&=&2^{-2n}\sum_{\tau}\sum_{\sigma}\prod_{(u,v)\in\mathcal{E}(H)}\left(\frac{p_{uv}^\varepsilon(\sigma^1)p_{uv}(\tau^1)}{p_0}\right)
\prod_{(u,v)\in\overline{\mathcal{E}(H)}}
\left(\frac{p_{uv}^\varepsilon(\sigma)p_{uv}(\tau)}{p_0}+\frac{q_{uv}^\varepsilon(\sigma)q_{uv}(\tau)}{q_0}\right)\nonumber\\
&=&\E_{\sigma\tau}\prod_{(u,v)\in\mathcal{E}(H)}\left(\frac{p_{uv}^\varepsilon(\sigma^1)p_{uv}(\tau^1)}{p_0}\right)
\prod_{(u,v)\in\overline{\mathcal{E}(H)}}
\left(\frac{p_{uv}^\varepsilon(\sigma)p_{uv}(\tau)}{p_0}+\frac{q_{uv}^\varepsilon(\sigma)q_{uv}(\tau)}{q_0}\right)\nonumber\\
&\equiv&\E_{\sigma\tau}X_H^\varepsilon(\sigma^1,\tau^1)W_H^\varepsilon(\sigma,\tau)Z_H^\varepsilon(\sigma^2,\tau^2),\label{thm:power:case2:eqn2}
\end{eqnarray}
where
\[
X_H^\varepsilon(\sigma^1,\tau^1)=\prod_{(u,v)\in\mathcal{E}(H)}\left(\frac{1}{p_0}p_{uv}^\varepsilon(\sigma^1)p_{uv}(\tau^1)\right),
\]
\[
W_H^\varepsilon(\sigma,\tau)=\prod_{(u,v)\in S_1(H)}
\left(\frac{1}{p_0}p_{uv}^\varepsilon(\sigma)p_{uv}(\tau)+\frac{1}{q_0}q_{uv}^\varepsilon(\sigma)q_{uv}(\tau)\right)
\]
and
\[
Z_H^\varepsilon(\sigma^2,\tau^2)=\prod_{(u,v)\in S_2(H)}
\left(\frac{1}{p_0}p_{uv}^\varepsilon(\sigma^2)p_{uv}(\tau^2)+\frac{1}{q_0}q_{uv}^\varepsilon(\sigma^2)q_{uv}(\tau^2)\right).
\]
Here $S_1(H)=\{(u,v)\in\overline{\mathcal{E}(H)}:\textrm{$u\in\mathcal{V}(H)$ or $v\in\mathcal{V}(H)$}\}$
and $S_2(H)=\{(u,v)\in\overline{\mathcal{E}(H)}:u,v\notin\mathcal{V}(H)\}$.

We will show that  $W_H^\varepsilon(\sigma,\tau)$ is uniformly bounded over $\sigma,\tau,H$, and that
\begin{equation}\label{thm:power:case2:eqn3}
\sup_{H\in \mathcal{H}}|W_H^\varepsilon(\sigma,\tau)-1|\to0,\,\,a.s.
\end{equation}
To see this, observe that
\begin{eqnarray}\label{thm:power:case2:eqn4}
W_H^\varepsilon(\sigma,\tau)&=&\prod_{(u,v)\in\overline{\mathcal{E}(H)},
u,v\in\mathcal{V}(H)}\left(\frac{1}{p_0}
p_{uv}^\varepsilon(\sigma)p_{uv}(\tau)+\frac{1}{q_0}q_{uv}^\varepsilon(\sigma)q_{uv}(\tau)\right)\nonumber\\
&&\times\prod_{v\in\mathcal{V}(H)}\prod_{u\notin\mathcal{V}(H)}
\left(\frac{1}{p_0}p_{uv}^\varepsilon(\sigma)p_{uv}(\tau)+\frac{1}{q_0}q_{uv}^\varepsilon(\sigma)q_{uv}(\tau)\right).
\end{eqnarray}
We note that
\[
\frac{1}{p_0}
p_{uv}^\varepsilon(\sigma)p_{uv}(\tau)+\frac{1}{q_0}q_{uv}^\varepsilon(\sigma)q_{uv}(\tau)=
1+O(n^{-1}),
\]
where the $O(n^{-1})$ term is uniform for $u,v,\sigma,\tau,H$.
The first product in (\ref{thm:power:case2:eqn4}) is therefore equal to $(1+O(n^{-1}))^{{M\choose2}-M}=1+o(1)$.
We turn to the second product in (\ref{thm:power:case2:eqn4}).
For any $v\in\mathcal{V}$, let
\begin{eqnarray*}
S_v^1&=&\#\{u\notin\mathcal{V}(H): \sigma_u^2=\sigma_v^1, \tau_u^2=\tau_v^1\}\\
S_v^2&=&\#\{u\notin\mathcal{V}(H): \sigma_u^2=\sigma_v^1, \tau_u^2\neq\tau_v^1\}\\
S_v^3&=&\#\{u\notin\mathcal{V}(H): \sigma_u^2\neq\sigma_v^1, \tau_u^2=\tau_v^1\}\\
S_v^4&=&\#\{u\notin\mathcal{V}(H): \sigma_u^2\neq\sigma_v^1, \tau_u^2\neq\tau_v^1\}.
\end{eqnarray*}
Also let
$S_{ll'}=\#\{u\notin\mathcal{V}(H): \sigma_u^2=l,\tau_u^2=l'\}$
and $N_{ll'}=\{v\in\mathcal{V}(H): \sigma_v^1=l, \tau_v^1=l'\}$
for $l,l'=\pm$.
Then the second product in (\ref{thm:power:case2:eqn4}) equals to
\begin{eqnarray*}
&&\prod_{v\in\mathcal{V}(H)}\prod_{u\notin\mathcal{V}(H)}
\left(\frac{1}{p_0}\left(\frac{a_\varepsilon}{n}\right)^{1(\sigma_u^2=\sigma_v^1)}
\left(\frac{b_\varepsilon}{n}\right)^{1(\sigma_u^2\neq\sigma_v^1)}
\left(\frac{a}{n}\right)^{1(\tau_u^2=\tau_v^1)}\left(\frac{b}{n}\right)^{1(\tau_u^2\neq\tau_v^1)}\right.\\
&&\left.+\frac{1}{q_0}\left(1-\frac{a_\varepsilon}{n}\right)^{1(\sigma_u^2=\sigma_v^1)}
\left(1-\frac{b_\varepsilon}{n}\right)^{1(\sigma_u^2\neq\sigma_v^1)}
\left(1-\frac{a}{n}\right)^{1(\tau_u^2=\tau_v^1)}\left(1-\frac{b}{n}\right)^{1(\tau_u^2\neq\tau_v^1)}\right)\\
&=&
\prod_{v\in\mathcal{V}(H)}
\left(\frac{1}{p_0}\left(\frac{a_\varepsilon}{n}\right)
\left(\frac{a}{n}\right)+\frac{1}{q_0}\left(1-\frac{a_\varepsilon}{n}\right)\left(1-\frac{a}{n}\right)\right)^{S_v^1}\\
&&\times\left(\frac{1}{p_0}\left(\frac{a_\varepsilon}{n}\right)\left(\frac{b}{n}\right)+\frac{1}{q_0}
\left(1-\frac{a_\varepsilon}{n}\right)\left(1-\frac{b}{n}\right)\right)^{S_v^2}\\
&&\times\left(\frac{1}{p_0}\left(\frac{b_\varepsilon}{n}\right)\left(\frac{a}{n}\right)+\frac{1}{q_0}
\left(1-\frac{b_\varepsilon}{n}\right)\left(1-\frac{a}{n}\right)\right)^{S_v^3}\\
&&\times\left(\frac{1}{p_0}\left(\frac{b_\varepsilon}{n}\right)\left(\frac{b}{n}\right)+\frac{1}{q_0}
\left(1-\frac{b_\varepsilon}{n}\right)\left(1-\frac{b}{n}\right)\right)^{S_v^4}\\
&=&(1+\widetilde{\gamma}_n^\varepsilon+O(n^{-3}))^{\sum_{v\in\mathcal{V}(H)}(S_v^1+S_v^4)}
(1-\widetilde{\gamma}_n^\varepsilon+O(n^{-3}))^{\sum_{v\in\mathcal{V}(H)}(S_v^2+S_v^3)}.
\end{eqnarray*}
In the above we have used the following trivial facts:
\begin{eqnarray*}
\frac{1}{p_0}\left(\frac{a_\varepsilon}{n}\right)
\left(\frac{a}{n}\right)+\frac{1}{q_0}\left(1-\frac{a_\varepsilon}{n}\right)\left(1-\frac{a}{n}\right)&=&1+\widetilde{\gamma}_n^\varepsilon+O(n^{-3})\\
\frac{1}{p_0}\left(\frac{a_\varepsilon}{n}\right)\left(\frac{b}{n}\right)+\frac{1}{q_0}
\left(1-\frac{a_\varepsilon}{n}\right)\left(1-\frac{b}{n}\right)&=&1-\widetilde{\gamma}_n^\varepsilon+O(n^{-3})\\
\frac{1}{p_0}\left(\frac{b_\varepsilon}{n}\right)\left(\frac{a}{n}\right)+\frac{1}{q_0}
\left(1-\frac{b_\varepsilon}{n}\right)\left(1-\frac{a}{n}\right)&=&1-\widetilde{\gamma}_n^\varepsilon+O(n^{-3})\\
\frac{1}{p_0}\left(\frac{b_\varepsilon}{n}\right)\left(\frac{b}{n}\right)+\frac{1}{q_0}
\left(1-\frac{b_\varepsilon}{n}\right)\left(1-\frac{b}{n}\right)&=&1+\widetilde{\gamma}_n^\varepsilon+O(n^{-3}),
\end{eqnarray*}
where $\widetilde{\gamma}_n^\varepsilon=\frac{\widetilde{\kappa}_\varepsilon}{n}+\frac{(a-b)(a_\varepsilon-b_\varepsilon)}{4n^2}$
and $\widetilde{\kappa}_\varepsilon=\frac{(a-b)(a_\varepsilon-b_\varepsilon)}{2(a+b)}$.
Note that
\begin{eqnarray*}
\sum_{v\in\mathcal{V}(H)}(S_v^1+S_v^4)&=&\sum_{\sigma_v^1=+,\tau_v^1=+}(S_{++}+S_{--})
+\sum_{\sigma_v^1=+,\tau_v^1=-}(S_{+-}+S_{-+})\\
&&+\sum_{\sigma_v^1=-,\tau_v^1=+}(S_{-+}+S_{+-})
+\sum_{\sigma_v^1=-,\tau_v^1=-}(S_{--}+S_{++})\\
&=&(S_{++}+S_{--})(N_{++}+N_{--})+(S_{+-}+S_{-+})(N_{+-}+N_{-+})\equiv N_1,
\end{eqnarray*}
similarly,
\[
\sum_{v\in\mathcal{V}(H)}(S_v^2+S_v^3)=(S_{+-}+S_{-+})(N_{++}+N_{--})+(S_{++}+S_{--})(N_{+-}+N_{-+})\equiv N_2.
\]
So the second product in (\ref{thm:power:case2:eqn4}) equals to
\begin{eqnarray*}
(1+o(1))(1+\widetilde{\gamma}_n^\varepsilon)^{N_1}(1-\widetilde{\gamma}_n^\varepsilon)^{N_2}=(1+o(1))\exp\left(\frac{N_1-N_2}{n}\widetilde{\kappa}_\varepsilon\right),
\end{eqnarray*}
where the $o(1)$ term is uniform for $u,v,\sigma,\tau,H$, thanks to $N_1, N_2\le Mn$.
By law of large number, $(N_1-N_2)/n\to0$, a.s., uniformly for $H\in\mathcal{H}$.
Therefore (\ref{thm:power:case2:eqn3}) holds. The above analysis also shows that
$W_H^\varepsilon(\sigma,\tau)$ is uniformly bounded over $\sigma,\tau,H$.

Next let us analyze the term $Z_H^\varepsilon(\sigma^2,\tau^2)$.
By Taylor expansions and direct examinations it can be checked that
for $u,v\in[n]\backslash\mathcal{V}(H)$, 
\[
\frac{1}{p_0}p_{uv}^\varepsilon(\sigma^2)p_{uv}(\tau^2)+\frac{1}{q_0}q_{uv}^\varepsilon(\sigma^2)q_{uv}(\tau^2)
=\left\{\begin{array}{cc}
1+\widetilde{\gamma}_n^\varepsilon+O(n^{-3}),&\textrm{if $\sigma_u^2\sigma_v^2\tau_u^2\tau_v^2=+$}\\
1-\widetilde{\gamma}_n^\varepsilon+O(n^{-3}),&\textrm{if $\sigma_u^2\sigma_v^2\tau_u^2\tau_v^2=-$}.
\end{array}\right.
\]
Let $s_+=\#\{(u,v): u,v\in[n]\backslash\mathcal{V}(H), u<v, \sigma_u^2\sigma_v^2\tau_u^2\tau_v^2=+\}$
and $s_-=\#\{(u,v): u,v\in[n]\backslash\mathcal{V}(H), u<v, \sigma_u^2\sigma_v^2\tau_u^2\tau_v^2=-\}$.
Let $\rho=\rho(\sigma^2,\tau^2)=\frac{1}{n-M}\sum_{u\in [n]\backslash\mathcal{V}(H)}\sigma_u^2\tau_u^2$.
By direct examinations we have
\begin{eqnarray}
Z_H^\varepsilon(\sigma^2,\tau^2)
&=&(1+\widetilde{\gamma}_n^\varepsilon+O(n^{-3}))^{s_+}(1-\widetilde{\gamma}_n^\varepsilon+O(n^{-3}))^{s_-}\nonumber\\
&=&(1+o(1))(1+\widetilde{\gamma}_n^\varepsilon)^{s_+}(1-\widetilde{\gamma}_n^\varepsilon)^{s_-}\nonumber\\
&=&(1+o(1))\exp\left(-\frac{\widetilde{\kappa}_\varepsilon^2(n-M)^2}{4n^2}-\frac{\widetilde{\kappa}_\varepsilon(n-M)}{2n}\right)\nonumber\\
&&\times\exp\left(\frac{(\sqrt{n-M}\rho)^2}{2}\left(\frac{(a-b)(a_\varepsilon-b_\varepsilon)(n-M)}{4n^2}+\frac{\widetilde{\kappa}_\varepsilon(n-M)}{n}\right)\right).\label{thm:power:case1:eqn5}
\end{eqnarray}

By the condition $(a-b)(a_\varepsilon-b_\varepsilon)<2(a+b)/3$, $\widetilde{\kappa}_\varepsilon<1$.
Let $Z_n=\sqrt{n-M}\rho$.
Let $\kappa_n=\frac{(a-b)(a_\varepsilon-b_\varepsilon)(n-M)}{4n^2}+\frac{\widetilde{\kappa}_\varepsilon(n-M)}{n}$ which
is nonrandom tending to $\widetilde{\kappa}_\varepsilon$. By Hoeffding's inequality:
for any $C>0$,
\begin{eqnarray}\label{thm:power:case1:eqn8}
P\left(\exp(\kappa_n Z_n^2/2)\ge C\right)\le2C^{-1/\kappa_n}.
\end{eqnarray}
From (\ref{thm:power:case1:eqn5}) there exists a universal constant $C_0$ such that
$Z_H^\varepsilon(\sigma^2,\tau^2)\le C_0\exp(\kappa_n Z_n^2/2)$,
hence, it follows from (\ref{thm:power:case1:eqn8}) that
for all $C>0$,
\[
P\left(Z_H^\varepsilon(\sigma^2,\tau^2)\ge C\right)\le2(C/C_0)^{-1/\kappa_n}.
\]
Therefore, by (\ref{thm:power:case1:eqn8}) we have that
\begin{eqnarray}
&&\E_{\sigma^2\tau^2}Z_H^\varepsilon(\sigma^2,\tau^2)1(Z_H^\varepsilon(\sigma^2,\tau^2)\ge C)\nonumber\\
&=&\int_0^\infty P\left(Z_H^\varepsilon(\sigma^2,\tau^2)1(Z_H^\varepsilon(\sigma^2,\tau^2)\ge C)>t\right)dt\nonumber\\
&=&CP\left(Z_H^\varepsilon(\sigma^2,\tau^2)\ge C\right)+\int_C^\infty
P\left(Z_H^\varepsilon(\sigma^2,\tau^2)>t\right)dt\nonumber\\
&\le&2C_0^{1/\kappa_n}C^{1-1/\kappa_n}/(1-\kappa_n).\label{thm:power:case1:eqn9}
\end{eqnarray}
We can also show that, as $n\to\infty$, 
\begin{equation}\label{thm:power:case2:eqn5}
\sup_{H\in\mathcal{H}}|\E_{\sigma^2\tau^2}Z_H^\varepsilon(\sigma^2,\tau^2)-\exp\left(-\widetilde{\kappa}_\varepsilon^2/4-\widetilde{\kappa}_\varepsilon/2\right)(1-\widetilde{\kappa}_\varepsilon)^{-1/2}|\to0,
n\to\infty.
\end{equation}
To see this, let $\rho_0=\frac{1}{n-M}\sum_{u\in[n]}\sigma_u\tau_u$ and $r_H=\frac{1}{n-M}\sum_{u\in\mathcal{V}(H)}\sigma_u\tau_u$,
therefore, $\rho=\rho_0-r_H$. Let $Z_{0n}=\sqrt{n-M}\rho_0$.
Then for any $H\in \mathcal{H}$, $|r_H|\le M/(n-M)$ which leads to
\[
Z_{0n}^2-2M|\rho_0|\le Z_n^2\le Z_{0n}^2-2M|\rho_0|+M^2/(n-M).
\]
Both left and right hand sides in the above are free of $H$ and converge to $\chi_1^2$ thanks to $\rho_0\to0$, a.s.
So
\[
\sup_{H\in \mathcal{H}}|\E\exp(\kappa_n Z_n^2/2)-(1-\widetilde{\kappa}_{\varepsilon})^{-1/2}|\to0,\,\,n\to\infty.
\]
This, together with (\ref{thm:power:case1:eqn5}), prove (\ref{thm:power:case2:eqn5}).

Next let us analyze $X_H^\varepsilon(\sigma^1,\tau^1)$.
Assume $H=(H_{mi})_{3\le m\le k, 1\le i\le j_m}\in A$.
For $3\le m\le k$ and $1\le i\le j_m$,
let $\tau^{1mi}$, $\sigma^{1mi}$ be the restrictions of $\tau^1$, $\sigma^1$ over the vertexes of $H_{mi}$.
Since $H_{mi}$ are vertex-disjoint, $\tau^{1mi}$'s, $\sigma^{1mi}$'s are all independent.
Let $N_{mi}=\sum_{(u,v)\in\mathcal{E}(H_{mi})}1(\sigma_u^{1mi}\neq\sigma^{1mi}_v)$,
the number of edges over $H_{mi}$ with distinct end points.
Following the proof of \cite[Lemma 3.3]{MNS15}, we get that
\begin{eqnarray}
&&\E_{\sigma^1\tau^1}X_H^\varepsilon(\sigma^1,\tau^1)\nonumber\\
&=&\E_{\sigma^1\tau^1}\prod_{(u,v)\in\mathcal{E}(H)}\left(\frac{1}{p_0}p_{uv}^\varepsilon(\sigma^1)p_{uv}(\tau^1)\right)\nonumber\\
&=&p_0^{-M}\prod_{m=3}^k\prod_{i=1}^{j_m}\E_{\sigma^{1}}\prod_{(u,v)\in\mathcal{E}(H_{mi})}p_{uv}^\varepsilon(\sigma^{1mi})\times\E_{\tau^1}
\prod_{(u,v)\in\mathcal{E}(H_{mi})}p_{uv}(\tau^{1mi})\nonumber\\
&=&p_0^{-M}n^{-2M}\prod_{m=3}^k\prod_{i=1}^{j_m}\E_{\sigma^{1mi}}a_\varepsilon^{m-N_{mi}}b_\varepsilon^{N_{mi}}\times \E_{\tau^{1mi}}a^{m-N_{mi}}b^{N_{mi}}\nonumber\\
&=&p_0^{-M}n^{-2M}\prod_{m=3}^k\prod_{i=1}^{j_m}2^{-m}\left[(a_\varepsilon+b_\varepsilon)^m+(a_\varepsilon-b_\varepsilon)^m\right]
\times 2^{-m}\left[(a+b)^m+(a-b)^m\right]\nonumber\\
&=&n^{-M}\prod_{m=3}^k\left[\left(\frac{a+b}{2}\right)^m(1+\delta_m)(1+\delta_m^\varepsilon)\right]^{j_m},\label{case2:XHe}
\end{eqnarray}
recalling $\delta_m=\left(\frac{a-b}{a+b}\right)^m$ and $\delta_m^\varepsilon=\left(\frac{a_\varepsilon-b_\varepsilon}{a+b}\right)^m$.
Meanwhile, it is easy to see that $n^MX_H^\varepsilon(\sigma^1,\tau^1)$ is almost surely bounded and the bound is unrelated to the vertexes of
$H$, i.e.,
\begin{equation}\label{thm:power:case1:eqn11}
n^M X_H^\varepsilon(\sigma^1,\tau^1)\le\left(\frac{2a_\varepsilon a}{a+b}\right)^M,\,\,\,\,\forall\sigma^1,\tau^1\in\{\pm\}^{\mathcal{V}(H)}.
\end{equation}

By (\ref{thm:power:case2:eqn3}), (\ref{thm:power:case1:eqn9}), (\ref{thm:power:case1:eqn11}), and bounded convergence theorem,
we can show that
\begin{equation}\label{thm:power:case1:eqn10}
\sum_{H\in A}\E_{\sigma\tau}X_H^\varepsilon(\sigma^1,\tau^1)|W_H^\varepsilon(\sigma,\tau)-1|Z_H^\varepsilon(\sigma^2,\tau^2)\to0.
\end{equation}
More precisely, using $|A|=\frac{n!}{(n-M)!}\prod_{m=3}^k\left(\frac{1}{2m}\right)^{j_m}$ (see proof of Lemma \ref{lemma:short:cycles}),
(\ref{thm:power:case1:eqn10}) follows from the following
\begin{eqnarray*}
&&\sum_{H\in A}\E_{\sigma\tau}X_H^\varepsilon(\sigma^1,\tau^1)|W_H^\varepsilon(\sigma,\tau)-1|Z_H^\varepsilon(\sigma^2,\tau^2)\\
&=&\sum_{H\in A}\E_{\sigma\tau}X_H^\varepsilon(\sigma^1,\tau^1)|W_H^\varepsilon(\sigma,\tau)-1|Z_H^\varepsilon(\sigma^2,\tau^2)1(Z_H^\varepsilon(\sigma^2,\tau^2)\le C)\\
&&+
\sum_{H\in A}\E_{\sigma\tau}X_H^\varepsilon(\sigma^1,\tau^1)|W_H^\varepsilon(\sigma,\tau)-1|Z_H^\varepsilon(\sigma^2,\tau^2)1(Z_H(\sigma^2,\tau^2)>C)\\
&\lesssim&Cn^{-M}|A|\E_{\sigma\tau}\sup_{H\in A}|W_H^\varepsilon(\sigma,\tau)-1|+n^{-M}|A|\sup_{H\in A}\E_{\sigma\tau}Z_H^\varepsilon(\sigma^2,\tau^2)1(Z_H^\varepsilon(\sigma^2,\tau^2)>C)\\
&\to&0,
\end{eqnarray*}
where the last limit follows by first taking $C\to\infty$ and then $n\to\infty$.
By (\ref{thm:power:case2:eqn2}), (\ref{thm:power:case2:eqn5}) and (\ref{case2:XHe}), we have that
\begin{eqnarray*}
&&\sum_{H\in A}\E_1 Y_n^\varepsilon 1_H\\
&=&|A|n^{-M}\prod_{m=3}^k\left[\left(\frac{a+b}{2}\right)^m(1+\delta_m^\varepsilon)(1+\delta_m)\right]^{j_m}\exp(-\widetilde{\kappa}_\varepsilon^2/4-\widetilde{\kappa}_\varepsilon/2)/\sqrt{1-\widetilde{\kappa}_\varepsilon}+o(1)\\
&\overset{n\to\infty}{\rightarrow}&
\prod_{m=3}^k(\lambda_m(1+\delta_m^\varepsilon)(1+\delta_m))^{j_m}\exp(-\widetilde{\kappa}_\varepsilon^2/4-\widetilde{\kappa}_\varepsilon/2)/\sqrt{1-\widetilde{\kappa}_\varepsilon},
\end{eqnarray*}
recalling $\lambda_m=\frac{1}{2m}\left(\frac{a+b}{2}\right)^m$.

From (\ref{thm:power:case2:eqn2}), the uniform boundedness of $\E_{\sigma^2\tau^2}Z_H^\varepsilon(\sigma^2,\tau^2)$
and the uniform boundedness of $W_H^\varepsilon(\sigma,\tau)$, and the independence of $\sigma^1,\tau^1,\sigma^2,\tau^2$ that, there exists a constant $C_1$ s.t.
for any $H\in\bar{A}$,
\begin{eqnarray*}
\E_1Y_n^\varepsilon 1_H\le C_1\E_{\sigma^1\tau^1}X_H^\varepsilon(\sigma^1,\tau^1).
\end{eqnarray*}
Also notice from the definition of $X_H^\varepsilon$ that
\begin{eqnarray*}
X_H^\varepsilon(\sigma^1,\tau^1)&=&p_0^{-|\mathcal{E}(H)|}\prod_{(u,v)\in\mathcal{E}(H)}
\left(\frac{a_\varepsilon}{n}\right)^{1(\sigma_u^1=\sigma_v^1)+1(\tau_u^1=\tau_v^1)}\left(\frac{b_\varepsilon}{n}\right)^{1(\sigma_u^1\neq\sigma_v^1)+1(\tau_u^1\neq\tau_v^1)}\\
&\le&n^{-2|\mathcal{E}(H)|}p_0^{-|\mathcal{E}(H)|}a_\varepsilon^{2|\mathcal{E}(H)|}=n^{-|\mathcal{E}(H)|}\left(\frac{2a_\varepsilon^2}{a+b}\right)^{|\mathcal{E}(H)|}.
\end{eqnarray*}
Since there are at most ${n\choose |\mathcal{V}(H)|}|\mathcal{V}(H)|!$ graphs isomorphic to $H$, and $|\mathcal{E}(H)|>|\mathcal{V}(H)|$ for $H\in\bar{A}$,
we get that, as $n\to\infty$,
\begin{eqnarray*}
\sum_{\textrm{$H'$ is isomorphic to $H$}}\E_1 Y_n^\varepsilon 1_H\le 
C_1\left(\frac{2a_\varepsilon^2}{a+b}\right)^{|\mathcal{E}(H)|}n^{-|\mathcal{E}(H)|}{n\choose |\mathcal{V}(H)|}|\mathcal{V}(H)|!\to0.
\end{eqnarray*}
Since there is a bounded number of isomorphism classes, 
we get that the second part of (\ref{thm:power:case2:eqn-1}) tends to zero as $n\to\infty$.

Hence, as $n\to\infty$, 
\[
\textrm{(\ref{thm:power:case2:eqn0}) $\to\prod_{m=3}^k(\lambda_m(1+\delta_m^\varepsilon)(1+\delta_m))^{j_m}\exp(-\widetilde{\kappa}_\varepsilon^2/4-\widetilde{\kappa}_\varepsilon/2)/\sqrt{1-\widetilde{\kappa}_\varepsilon}$.}
\]
As for $\E_1 Y_n^\varepsilon$, note that it is
equal to
\[
\E_1 Y_n^\varepsilon=4^{-n}\sum_{\sigma,\tau}\prod_{u<v}\left(\frac{1}{p_0}p_{uv}^\varepsilon(\sigma)p_{uv}(\tau)
+\frac{1}{q_0}q_{uv}^\varepsilon(\sigma)q_{uv}(\tau)\right).
\]
Similar to (\ref{thm:power:case2:eqn5}), i.e., taking $H$ therein as empty graph, one gets that
\begin{equation}\label{thm:power:case2:eqn7}
\E_1 Y_n^\varepsilon\to\exp(-\widetilde{\kappa}_\varepsilon^2/4-\widetilde{\kappa}_\varepsilon/2)/\sqrt{1-\widetilde{\kappa}_\varepsilon}.
\end{equation}
Hence,
\[
\frac{\E_1Y_n^\varepsilon [X_{3n}]_{j_3}\cdots[X_{kn}]_{j_k}}{\E_1 Y_n^\varepsilon}\overset{n\to\infty}{\rightarrow}\prod_{m=3}^k(\lambda_m(1+\delta_m^\varepsilon)(1+\delta_m))^{j_m}.
\]
This verifies Condition A2.

\textbf{Check Condition A3}.
Since $\frac{\lambda_m(1+\delta_m^\varepsilon)(1+\delta_m)}{\lambda_m(1+\delta_m)}-1=\delta_m^\varepsilon$, and by (\ref{vareps:condition}),
we have 
\[
\sum_{m\ge3}\lambda_m(1+\delta_m)(\delta_m^\varepsilon)^2=\sum_{m=3}^\infty\frac{1}{2m}\left(\frac{(a_\varepsilon-b_\varepsilon)^2}{2(a+b)}\right)^m
+\sum_{m=3}^\infty\frac{1}{2m}\left(\frac{(a_\varepsilon-b_\varepsilon)^2(a-b)}{2(a+b)^2}\right)^m<\infty.
\]

\textbf{Check Condition A4}.
By direct examinations it can be checked that
\begin{eqnarray*}
\E_1\{(Y_n^\varepsilon)^2|\tau\}&=&4^{-n}\sum_{\sigma}\sum_{\eta}\E_1\left\{\prod_{u<v}
\left(\frac{p_{uv}^\varepsilon(\sigma)p_{uv}^\varepsilon(\eta)}{p_0^2}\right)^{A_{uv}}
\left(\frac{q_{uv}^\varepsilon(\sigma)q_{uv}^\varepsilon(\eta)}{q_0^2}\right)^{1-A_{uv}}\bigg|\tau\right\}\\
&=&4^{-n}\sum_{\sigma}\sum_{\eta}\prod_{u<v}\left(\frac{1}{p_0^2}p_{uv}^\varepsilon(\sigma)p_{uv}^\varepsilon(\eta)p_{uv}(\tau)+\frac{1}{q_0^2}q_{uv}^\varepsilon(\sigma)
q_{uv}^\varepsilon(\eta)q_{uv}(\tau)\right)
\end{eqnarray*}
So 
\begin{eqnarray}\label{thm:power:case1:eqn26}
\E_1(Y_n^\varepsilon)^2&=&8^{-n}\sum_{\sigma}\sum_{\eta}\sum_{\tau}
\prod_{u<v}\left(\frac{1}{p_0^2}p_{uv}^\varepsilon(\sigma)p_{uv}^\varepsilon(\eta)p_{uv}(\tau)+
\frac{1}{q_0^2}q_{uv}^\varepsilon(\sigma)q_{uv}^\varepsilon(\eta)q_{uv}(\tau)\right)\nonumber\\
&=&\E_{\sigma\eta\tau}\prod_{u<v}\left(\frac{1}{p_0^2}p_{uv}^\varepsilon(\sigma)p_{uv}^\varepsilon(\eta)p_{uv}(\tau)+
\frac{1}{q_0^2}q_{uv}^\varepsilon(\sigma)q_{uv}^\varepsilon(\eta)q_{uv}(\tau)\right),
\end{eqnarray}
where $\sigma,\eta,\tau$ in the above expectation are independent and uniformly distributed over $\{\pm\}^n$.
By Taylor expansion and straightforward (but exhaustive) calculations, it can be shown that
\begin{eqnarray}\label{def:gamma2+:gamma0-}
\frac{1}{p_0^2}\left(\frac{a_\varepsilon}{n}\right)^2\left(\frac{a}{n}\right)
+\frac{1}{q_0^2}\left(1-\frac{a_\varepsilon}{n}\right)^2\left(1-\frac{a}{n}\right)&=&1+\gamma_{2+}+O(n^{-3})\nonumber\\
\frac{1}{p_0^2}\left(\frac{a_\varepsilon}{n}\right)^2\left(\frac{b}{n}\right)
+\frac{1}{q_0^2}\left(1-\frac{a_\varepsilon}{n}\right)^2\left(1-\frac{b}{n}\right)&=&1+\gamma_{2-}+O(n^{-3})\nonumber\\
\frac{1}{p_0^2}\left(\frac{a_\varepsilon}{n}\right)\left(\frac{b_\varepsilon}{n}\right)\left(\frac{a}{n}\right)
+\frac{1}{q_0^2}\left(1-\frac{a_\varepsilon}{n}\right)\left(1-\frac{b_\varepsilon}{n}\right)\left(1-\frac{a}{n}\right)&=&1+\gamma_{1+}+O(n^{-3})\nonumber\\
\frac{1}{p_0^2}\left(\frac{a_\varepsilon}{n}\right)\left(\frac{b_\varepsilon}{n}\right)\left(\frac{b}{n}\right)
+\frac{1}{q_0^2}\left(1-\frac{a_\varepsilon}{n}\right)\left(1-\frac{b_\varepsilon}{n}\right)\left(1-\frac{b}{n}\right)&=&1+\gamma_{1-}+O(n^{-3})\nonumber\\
\frac{1}{p_0^2}\left(\frac{b_\varepsilon}{n}\right)^2\left(\frac{a}{n}\right)
+\frac{1}{q_0^2}\left(1-\frac{b_\varepsilon}{n}\right)^2\left(1-\frac{a}{n}\right)&=&1+\gamma_{0+}+O(n^{-3})\nonumber\\
\frac{1}{p_0^2}\left(\frac{b_\varepsilon}{n}\right)^2\left(\frac{b}{n}\right)
+\frac{1}{q_0^2}\left(1-\frac{b_\varepsilon}{n}\right)^2\left(1-\frac{b}{n}\right)&=&1+\gamma_{0-}+O(n^{-3})
\end{eqnarray}
where 
\begin{eqnarray*}
\gamma_{2+}&=&\frac{(a_\varepsilon-b_\varepsilon)^2(2a_\varepsilon+b_\varepsilon)+x_\varepsilon}{n(a+b)^2}+
\frac{3(a_\varepsilon-b_\varepsilon)^2+y_\varepsilon}{4n^2}\\
\gamma_{2-}&=&-\frac{a_\varepsilon(a_\varepsilon-b_\varepsilon)^2+x_\varepsilon}{n(a+b)^2}-\frac{(a_\varepsilon-b_\varepsilon)^2+y_\varepsilon}{4n^2}\\
\gamma_{1+}&=&-\frac{a_\varepsilon(a_\varepsilon-b_\varepsilon)^2+z_\varepsilon}{n(a+b)^2}-\frac{(a_\varepsilon-b_\varepsilon)^2}{4n^2}\\
\gamma_{1-}&=&-\frac{b_\varepsilon(a_\varepsilon-b_\varepsilon)^2-z_\varepsilon}{n(a+b)^2}-\frac{(a_\varepsilon-b_\varepsilon)^2}{4n^2}\\
\gamma_{0+}&=&-\frac{b_\varepsilon(a_\varepsilon-b_\varepsilon)^2+w_\varepsilon}{n(a+b)^2}-\frac{(a_\varepsilon-b_\varepsilon)^2+y_\varepsilon}{4n^2}\\
\gamma_{0-}&=&\frac{(a_\varepsilon-b_\varepsilon)^2(a_\varepsilon+2b_\varepsilon)+w_\varepsilon}{n(a+b)^2}+\frac{3(a_\varepsilon-b_\varepsilon)^2+y_\varepsilon}{4n^2}
\end{eqnarray*}
with
\[
x_\varepsilon=\varepsilon(a_\varepsilon-b_\varepsilon)(3a_\varepsilon+b_\varepsilon),
y_\varepsilon=4\varepsilon(a_\varepsilon-b_\varepsilon),
z_\varepsilon=\varepsilon(a_\varepsilon-b_\varepsilon)^2,
w_\varepsilon=\varepsilon(a_\varepsilon-b_\varepsilon)(a_\varepsilon+3b_\varepsilon).
\]
Define
$s_{r+}=\#\{(u,v): u<v, N_{uv}^{\sigma\tau}=r,\tau_u\tau_v=+\}$
and $s_{r-}=\#\{(u,v): u<v, N_{uv}^{\sigma\tau}=r,\tau_u\tau_v=-\}$,
for $r=0,1,2$.
Then it holds that
\begin{eqnarray}\label{thm:power:case2:eqn6}
\E_1(Y_n^\varepsilon)^2&=&\E_{\sigma\eta\tau}\prod_{r=0,1,2}(1+\gamma_{r+}+O(n^{-3}))^{s_{r+}}
\times\prod_{r=0,1,2}(1+\gamma_{r-}+O(n^{-3}))^{s_{r-}}\nonumber\\
&=&(1+o(1))\E_{\sigma\eta\tau}\prod_{r=0,1,2}(1+\gamma_{r+})^{s_{r+}}
\times\prod_{r=0,1,2}(1+\gamma_{r-})^{s_{r-}}.
\end{eqnarray}
Define 
\begin{eqnarray}\label{def:rhos}
\rho_1&=&\frac{1}{\sqrt{n}}\sum_{u=1}^n\sigma_u,\rho_2=\frac{1}{\sqrt{n}}\sum_{u=1}^n\eta_u,\rho_3=\frac{1}{\sqrt{n}}\sum_{u=1}^n\tau_u,\nonumber\\
\rho_4&=&\frac{1}{\sqrt{n}}\sum_{u=1}^n\sigma_u\eta_u,\rho_5=\frac{1}{\sqrt{n}}\sum_{u=1}^n\sigma_u\tau_u,
\rho_6=\frac{1}{\sqrt{n}}\sum_{u=1}^n\eta_u\tau_u,
\rho_7=\frac{1}{\sqrt{n}}\sum_{u=1}^n\sigma_u\eta_u\tau_u.
\end{eqnarray}

Observe that
\begin{eqnarray*}
s_{2+}&=&\sum_{u<v}1(\sigma_u\sigma_v=+)1(\eta_u\eta_v=+)1(\tau_u\tau_v=+)\\
&=&\frac{n^2}{16}-\frac{n}{2}+\frac{n}{16}\left(\rho_1^2+\rho_2^2+\rho_3^2+\rho_4^2+\rho_5^2+\rho_6^2+\rho_7^2\right)\\
s_{2-}&=&\sum_{u<v}1(\sigma_u\sigma_v=+)1(\eta_u\eta_v=+)1(\tau_u\tau_v=-)\\
&=&\frac{n^2}{16}+\frac{n}{16}\left(\rho_1^2+\rho_2^2-\rho_3^2+\rho_4^2-\rho_5^2-\rho_6^2-\rho_7^2\right)
\end{eqnarray*}
\begin{eqnarray*}
s_{1+}&=&\sum_{u<v}1(\sigma_u\sigma_v\eta_u\eta_v=-)1(\tau_u\tau_v=+)
=\frac{n^2}{8}+\frac{n}{8}\left(\rho_3^2-\rho_4^2-\rho_7^2\right)\\
s_{1-}&=&\sum_{u<v}1(\sigma_u\sigma_v\eta_u\eta_v=-)1(\tau_u\tau_v=-)
=\frac{n^2}{8}-\frac{n}{8}\left(\rho_3^2+\rho_4^2-\rho_7^2\right)
\end{eqnarray*}
\begin{eqnarray*}
s_{0+}&=&\sum_{u<v}1(\sigma_u\sigma_v=-)1(\eta_u\eta_v=-)1(\tau_u\tau_v=+)\\
&=&\frac{n^2}{16}+\frac{n}{16}\left(-\rho_1^2-\rho_2^2+\rho_3^2+\rho_4^2-\rho_5^2-\rho_6^2+\rho_7^2\right)\\
s_{0-}&=&\sum_{u<v}1(\sigma_u\sigma_v=-)1(\eta_u\eta_v=-)1(\tau_u\tau_v=-)\\
&=&\frac{n^2}{16}+\frac{n}{16}\left(-\rho_1^2-\rho_2^2-\rho_3^2+\rho_4^2+\rho_5^2+\rho_6^2-\rho_7^2\right).
\end{eqnarray*}
Using the above notation $\gamma_{r\pm}$'s and $s_{r\pm}$'s we can write the right hand side of (\ref{thm:power:case2:eqn6}) as
\[
\prod_{r=0,1,2}(1+\gamma_{r+})^{s_{r+}}
\times\prod_{r=0,1,2}(1+\gamma_{r-})^{s_{r-}}\equiv T_1\times T_2,
\]
where
\begin{eqnarray*}
T_1&=&(1+\gamma_{2+})^{\frac{n^2}{16}-\frac{n}{2}}
(1+\gamma_{2-})^{\frac{n^2}{16}}(1+\gamma_{1+})^{\frac{n^2}{8}}
(1+\gamma_{1-})^{\frac{n^2}{8}}(1+\gamma_{0+})^{\frac{n^2}{16}}(1+\gamma_{0-})^{\frac{n^2}{16}}\\
&=&(1+o(1))\exp\left(-\frac{(a_\varepsilon-b_\varepsilon)^2}{4(a+b)^4}
\left[(a_\varepsilon-b_\varepsilon)^2(a_\varepsilon^2+a_\varepsilon b_\varepsilon+b_\varepsilon^2)
+\varepsilon(a_\varepsilon-b_\varepsilon)(3a_\varepsilon^2+2a_\varepsilon
b_\varepsilon+3b_\varepsilon^2)\right.\right.\\
&&\left.\left.+\varepsilon^2(3a_\varepsilon^2+2a_\varepsilon
b_\varepsilon+3b_\varepsilon^2)\right]-\frac{(a_\varepsilon-b_\varepsilon)^2(2a_\varepsilon+b_\varepsilon)+x_\varepsilon}{2(a+b)^2}\right)\\
&=&(1+o(1))\exp\left(-\frac{(a_\varepsilon-b_\varepsilon)^4}{16(a+b)^2}-\frac{(a_\varepsilon-b_\varepsilon)^4(a-b)^2}{16(a+b)^4}-\frac{(a_\varepsilon-b_\varepsilon)^2(a-b)^2}{8(a+b)^2}\right.\\
&&\left.-\frac{(a_\varepsilon-b_\varepsilon)^2(a-b)}{4(a+b)^2}-\frac{(a_\varepsilon-b_\varepsilon)^2}{4(a+b)}
-\frac{(a-b)(a_\varepsilon-b_\varepsilon)}{2(a+b)}\right)\\
&=&(1+o(1))\exp\left(-\kappa_\varepsilon^2/4-\kappa_\varepsilon/2\right)\exp\left(-\widetilde{\kappa}_\varepsilon^2/2-\widetilde{\kappa}_\varepsilon\right)\exp\left(-\frac{(a_\varepsilon-b_\varepsilon)^4(a-b)^2}{16(a+b)^4}-\frac{(a_\varepsilon-b_\varepsilon)^2(a-b)}{4(a+b)^2}\right),
\end{eqnarray*}
and 
\begin{eqnarray*}
T_2&=&(1+o(1))(1+\gamma_{2+})^{\frac{n}{16}\left(\rho_1^2+\rho_2^2+\rho_3^2+\rho_4^2+\rho_5^2+\rho_6^2+\rho_7^2\right)}
(1+\gamma_{2-})^{\frac{n}{16}\left(\rho_1^2+\rho_2^2-\rho_3^2+\rho_4^2-\rho_5^2-\rho_6^2-\rho_7^2\right)}\\
&&\times(1+\gamma_{1+})^{\frac{n}{8}\left(\rho_3^2-\rho_4^2-\rho_7^2\right)}
(1+\gamma_{1-})^{-\frac{n}{8}\left(\rho_3^2+\rho_4^2-\rho_7^2\right)}\\
&&\times
(1+\gamma_{0+})^{\frac{n}{16}\left(-\rho_1^2-\rho_2^2+\rho_3^2+\rho_4^2-\rho_5^2-\rho_6^2+\rho_7^2\right)}
(1+\gamma_{0-})^{\frac{n}{16}\left(-\rho_1^2-\rho_2^2-\rho_3^2+\rho_4^2+\rho_5^2+\rho_6^2-\rho_7^2\right)}\\
&=&(1+o(1))\exp\left(\frac{(a_\varepsilon-b_\varepsilon)^2(2a_\varepsilon+b_\varepsilon)+x_\varepsilon}{16(a+b)^2}
\left(\rho_1^2+\rho_2^2+\rho_3^2+\rho_4^2+\rho_5^2+\rho_6^2+\rho_7^2\right)\right.\\
&&\left.-\frac{a_\varepsilon(a_\varepsilon-b_\varepsilon)^2+x_\varepsilon}{16(a+b)^2}
\left(\rho_1^2+\rho_2^2-\rho_3^2+\rho_4^2-\rho_5^2-\rho_6^2-\rho_7^2\right)\right.\\
&&\left.-\frac{a_\varepsilon(a_\varepsilon-b_\varepsilon)^2+z_\varepsilon}{8(a+b)^2}\left(\rho_3^2-\rho_4^2-\rho_7^2\right)\right.\\
&&\left.+\frac{b_\varepsilon(a_\varepsilon-b_\varepsilon)^2-z_\varepsilon}{8(a+b)^2}
\left(\rho_3^2+\rho_4^2-\rho_7^2\right)\right.\\
&&\left.+\frac{b_\varepsilon(a_\varepsilon-b_\varepsilon)^2+w_\varepsilon}{16(a+b)^2}\left(\rho_1^2+\rho_2^2-\rho_3^2-\rho_4^2+\rho_5^2+\rho_6^2-\rho_7^2\right)\right.\\
&&\left.-\frac{(a_\varepsilon-b_\varepsilon)^2(a_\varepsilon+2b_\varepsilon)+w_\varepsilon}{16(a+b)^2}\left(\rho_1^2+\rho_2^2+\rho_3^2-\rho_4^2-\rho_5^2-\rho_6^2+\rho_7^2\right)
\right)\\
&=&(1+o(1))\exp\left(\frac{\kappa_\varepsilon}{2}\rho_4^2+\frac{\widetilde{\kappa}_\varepsilon}{2}\left(\rho_5^2+\rho_6^2\right)+\frac{(a_\varepsilon-b_\varepsilon)^2(a-b)}{4(a+b)^2}\rho_7^2\right).
\end{eqnarray*}
We note that $\rho_7$ is independent of $(\rho_4,\rho_5,\rho_6)$, and the condition $\kappa_\varepsilon<\widetilde{\kappa}_\varepsilon\in(0,1/3)$
leads to uniform integrability of $\exp\left(\frac{\kappa_\varepsilon}{2}\rho_4^2+\frac{\widetilde{\kappa}_\varepsilon}{2}\left(\rho_5^2+\rho_6^2\right)\right)$,
and $\rho_4,\rho_5,\rho_6,\rho_7$ jointly converge in distribution to independent standard normal variables. Therefore,
we have that
\begin{eqnarray*}
&&\E_1(Y_n^\varepsilon)^2\\
&\overset{n\to\infty}{\rightarrow}&
\exp\left(-\kappa_\varepsilon^2/4-\kappa_\varepsilon/2\right)\exp\left(-\widetilde{\kappa}_\varepsilon^2/2-\widetilde{\kappa}_\varepsilon\right)\exp\left(-\frac{(a_\varepsilon-b_\varepsilon)^4(a-b)^2}{16(a+b)^4}-\frac{(a_\varepsilon-b_\varepsilon)^2(a-b)}{4(a+b)^2}\right)\\
&&\times(1-\kappa_\varepsilon)^{-1/2}(1-\widetilde{\kappa}_\varepsilon)^{-1}\left(1-\frac{(a_\varepsilon-b_\varepsilon)^2(a-b)}{2(a+b)}\right)^{-1/2}.
\end{eqnarray*}
By (\ref{thm:power:case2:eqn7}) we get that
\begin{eqnarray}\label{thm:power:case2:eqn8}
&&\frac{\E_1(Y_n^\varepsilon)^2}{(\E_1Y_n^\varepsilon)^2}\nonumber\\
&\overset{n\to\infty}{\rightarrow}&\exp\left(-\kappa_\varepsilon^2/4-\kappa_\varepsilon/2\right)(1-\kappa_\varepsilon)^{-1/2}\nonumber\\
&&\times\exp\left(-\frac{(a_\varepsilon-b_\varepsilon)^4(a-b)^2}{16(a+b)^4}-\frac{(a_\varepsilon-b_\varepsilon)^2(a-b)}{4(a+b)^2}\right)\left(1-\frac{(a_\varepsilon-b_\varepsilon)^2(a-b)}{2(a+b)}\right)^{-1/2}\nonumber\\
&=&\exp\left(\sum_{m=3}^\infty\frac{1}{2m}\left(\frac{(a_\varepsilon-b_\varepsilon)^2}{2(a+b)}\right)^m\right)
\times\exp\left(\sum_{m=3}^\infty\frac{1}{2m}\left(\frac{(a_\varepsilon-b_\varepsilon)^2(a-b)}{2(a+b)^2}\right)^m\right)\\
&=&\exp\left(\sum_{m=3}^\infty\lambda_m(1+\delta_m)(\delta_m^\varepsilon)^2\right),\nonumber
\end{eqnarray}
where (\ref{thm:power:case2:eqn8}) follows from the below trivial facts:
\begin{eqnarray*}
\exp\left(\sum_{m=3}^\infty\frac{1}{2m}\left(\frac{(a_\varepsilon-b_\varepsilon)^2}{2(a+b)}\right)^m\right)&=&
\exp\left(-\kappa_\varepsilon^2/4-\kappa_\varepsilon/2\right)(1-\kappa_\varepsilon)^{-1/2}\\
\exp\left(\sum_{m=3}^\infty\frac{1}{2m}\left(\frac{(a_\varepsilon-b_\varepsilon)^2(a-b)}{2(a+b)^2}\right)^m\right)
&=&\exp\left(-\frac{(a_\varepsilon-b_\varepsilon)^4(a-b)^2}{16(a+b)^4}-\frac{(a_\varepsilon-b_\varepsilon)^2(a-b)}{4(a+b)^2}\right)\\
&&\times\left(1-\frac{(a_\varepsilon-b_\varepsilon)^2(a-b)}{2(a+b)}\right)^{-1/2}.
\end{eqnarray*}
This verifies Condition A4.

In the end, notice that by (\ref{thm:power:case2:eqn7}),
\begin{eqnarray*}
\E_1Y_n^\varepsilon&\overset{n\to\infty}{\rightarrow}&\exp\left(-\widetilde{\kappa}_\varepsilon^2/4-\widetilde{\kappa}_\varepsilon/2\right)(1-\widetilde{\kappa}_\varepsilon)^{-1/2}\\
&=&\exp\left(\sum_{m=3}^\infty\frac{1}{2m}\widetilde{\kappa}_\varepsilon^m\right)\\
&=&\exp\left(\sum_{m=3}^\infty\frac{1}{2m}\left(\frac{(a_\varepsilon-b_\varepsilon)(a-b)}{2(a+b)}\right)^m\right)
=\exp\left(\sum_{m=3}^\infty\lambda_m\delta_m\delta_m^\varepsilon\right).
\end{eqnarray*}
The it follows by Proposition \ref{basic:prop} that
\[
Y_n^\varepsilon\overset{n\to\infty}{\rightarrow}\exp\left(\sum_{m=3}^\infty\lambda_m\delta_m\delta_m^\varepsilon\right)
\prod_{m=3}^\infty(1+\delta_m^\varepsilon)^{Z_m^1}\exp\left(-\lambda_m(1+\delta_m)\delta_m^\varepsilon\right)
=W_1^\varepsilon,
\]
where $Z_m^1$ are independent Poisson variable with mean $\lambda_m(1+\delta_m)$.
\end{proof}

\subsection{Proofs in Section \ref{sec:power:analysis}}
Before proofs, we need the following technical lemma.
\begin{Lemma}\label{lemma:series:convergence}
Suppose that $\{c_{ml}\}_{m,l=1}^\infty$ is a real sequence satisfying (1) $\lim\limits_{M\to\infty}\lim\limits_{l\to\infty}\sum_{m=M}^\infty c_{ml}^2=0$, and
(2) for any $m\ge 1$, $\lim\limits_{l\to\infty}c_{ml}=c_m$. Furthermore, for any $l\ge1$,
$\{N_{ml}\}_{m=1}^\infty$ are independent random variables of zero mean and unit variance, and
for any $m\ge1$, $N_{ml}\overset{d}{\to}N(0,1)$ as $l\to\infty$.
Then, as $l\to\infty$, $\sum_{m=1}^\infty c_{ml}N_{ml}\overset{d}{\to}N(0,\sum_{m=1}^\infty c_m^2)$. 
\end{Lemma}
\begin{proof}[Proof of Lemma \ref{lemma:series:convergence}]
Notice that $c_m$ is a square summable sequence. To see this,
note that for any $M<N$,
\[
\sum_{m=M}^N c_m^2=\lim\limits_{l\to\infty}\sum_{m=M}^N c_{ml}^2\le\lim\limits_{l\to\infty}\sum_{m=M}^\infty c_{ml}^2,
\]
and hence, taking $N\to\infty$ on the left side we have,
\[
\sum_{m=M}^\infty c_m^2\le\lim\limits_{l\to\infty}\sum_{m=M}^\infty c_{ml}^2,
\]
leading to $\lim_{M\to\infty}\sum_{m=M}^\infty c_m^2\le\lim\limits_{M\to\infty}\lim\limits_{l\to\infty}\sum_{m=M}^\infty c_{ml}^2=0$; see (1).
Hence $\sum_{m=1}^\infty c_m^2<\infty$.

For arbitrary $M$ and $\delta>0$, define an event $\mathcal{E}_{Ml}=\{|\sum_{m=M}^\infty c_{ml}N_{ml}|<\delta\}$.
Since $E|\sum_{m=M}^\infty c_{ml}N_{ml}|^2=\sum_{m=M}^\infty c_{ml}^2$, by condition (1) we can choose $l$ and $M$ large so that
$E|\sum_{m=M}^\infty c_{ml}N_{ml}|^2\le\delta^3$, and so $P(\mathcal{E}_{Ml})\ge 1-\delta$ by Chebyshev inequality.
By independence and asymptotic normality of $N_{ml}$ for $1\le m\le M-1$, and condition (2), one has
$\sum_{m=1}^{M-1}c_{ml}N_{ml}\overset{d}{\to}N(0,\sum_{m=1}^{M-1}c_m^2)$ as $l\to\infty$.
Define $T_l=\sum_{m=1}^\infty c_{ml}N_{ml}$.
Hence, for any $z\in\mathbb{R}$,
\begin{eqnarray*}
P\left(T_l\le z\right)&\le&P(T_l\le z,\mathcal{E}_{Ml})+\delta\\
&\le&P\left(\sum_{m=1}^{M-1}c_{ml}N_{ml}\le z+\delta\right)+\delta\overset{l\to\infty}{\to}\Phi\left(\frac{z+\delta}{\sqrt{\sum_{m=1}^{M-1}c_m^2}}\right)+\delta.
\end{eqnarray*}
Taking $\delta\to0$ and $M\to\infty$ in the above, we have $\limsup\limits_{l\to\infty}P(T_l\le z)\le \Phi\left(\frac{z}{\sqrt{\sum_{m=1}^\infty c_m^2}}\right)$.
Likewise one can show that $\liminf\limits_{l\to\infty}P(T_l\le z)\ge \Phi\left(\frac{z}{\sqrt{\sum_{m=1}^\infty c_m^2}}\right)$.
Then we have $\lim\limits_{l\to\infty}P(T_l\le z)=\Phi\left(\frac{z}{\sqrt{\sum_{m=1}^\infty c_m^2}}\right)$.
Proof completed.
\end{proof}

\begin{proof}[Proof of Theorem \ref{lim:power:case:2}]
The proof follows by Lemma \ref{lemma:series:convergence}.
We will analyze the distributions of $W_0^\varepsilon$ and $W_1^\varepsilon$.
Define $\Delta_\varepsilon=\sum_{m=3}^\infty\lambda_m\left(\log(1+\delta_m^\varepsilon)-\delta_m^\varepsilon\right)$.
Since, as $a+b\to\infty$,
\[
\sqrt{\lambda_m}\log(1+\delta_m^\varepsilon)\to\sqrt{\frac{1}{2m}k_1^m},
\sqrt{\lambda_m(1+\delta_m)}\log(1+\delta_m^\varepsilon)\to\sqrt{\frac{1}{2m}k_1^m},
\lambda_m\delta_m\log(1+\delta_m^\varepsilon)\to\frac{1}{2m}k_2^m,
\]
and
\[
\frac{Z_m^0-\lambda_m}{\sqrt{\lambda_m}}\overset{d}{\to}N(0,1),
\frac{Z_m^1-\lambda_m(1+\delta_m)}{\sqrt{\lambda_m(1+\delta_m)}}\overset{d}{\to}N(0,1).
\]
Therefore, by Lemma \ref{lemma:series:convergence} we have, as $a+b\to\infty$,
\[
\log{W_0^\varepsilon}-\Delta_\varepsilon=
\sum_{m=3}^\infty\frac{Z_m^0-\lambda_m}{\sqrt{\lambda_m}}\times \sqrt{\lambda_m}\log(1+\delta_m^\varepsilon)\overset{d}{\to}N(0,\sigma_1^2),
\]
and
\begin{eqnarray*}
\log{W_1^\varepsilon}-\Delta_\varepsilon&=&\sum_{m=3}^\infty\frac{Z_m^1-\lambda_m(1+\delta_m)}{\sqrt{\lambda_m(1+\delta_m)}}\times\sqrt{\lambda_m(1+\delta_m)}
\log(1+\delta_m^\varepsilon)+\sum_{m=3}^\infty\lambda_m\delta_m\log(1+\delta_m^\varepsilon)\\
&\overset{d}{\to}&N(\sigma_2^2,\sigma_1^2).
\end{eqnarray*}
Therefore, as $a+b\to\infty$,
\begin{eqnarray*}
1-\alpha=P(W_0^\varepsilon\le w_\alpha^\varepsilon)=
P\left(\frac{\log{W_0^\varepsilon}-\Delta_\varepsilon}{\sigma_1}\le\frac{\log{w_\alpha^\varepsilon}-\Delta_\varepsilon}{\sigma_1}\right),
\end{eqnarray*}
which implies $\frac{\log{w_\alpha^\varepsilon}-\Delta_\varepsilon}{\sigma_1}\to z_{1-\alpha}$, and hence,
\begin{eqnarray*}
P(a,b,\varepsilon)=P(W_1^\varepsilon\ge w_\alpha^\varepsilon)=
P\left(\frac{\log{W_1}^\varepsilon-\Delta_\varepsilon}{\sigma_1}\ge\frac{\log{w_\alpha^\varepsilon}-\Delta_\varepsilon}{\sigma_1}\right)
\to\Phi\left(\frac{\sigma_2^2}{\sigma_1}-z_{1-\alpha}\right).
\end{eqnarray*}
Proof completed.
\end{proof}
\subsection{Proofs in Section \ref{sec:approx}}

\begin{proof}[Proof of Theorem \ref{thm:approx2}]
Observe that
\[
Var\left(\frac{1}{M}\sum_{l=1}^M g_n^\varepsilon(\sigma[l])\bigg| A\right)=\frac{1}{M}\left[
\E_\sigma\left\{g_n^\varepsilon(\sigma)^2\big|A\right\}-\E_\sigma\left\{g_n^\varepsilon(\sigma)\big|A\right\}^2\right]\le\frac{1}{M}
\E_\sigma\left\{g_n^\varepsilon(\sigma)^2\big|A\right\},
\]
where the variance is taken w.r.t. $\sigma[l]$'s conditional on $A_{uv}$'s.
So it is sufficient to deal with $\E_{A,\sigma}g_n^\varepsilon(\sigma)^2$.
First, assume $H_0$ holds. Then it holds that 
\[
\E_{A,\sigma}g_n^\varepsilon(\sigma)^2=\E_\sigma\prod_{u<v}\left(\frac{p_{uv}^\varepsilon(\sigma)^2}{p_0}+\frac{q_{uv}^\varepsilon(\sigma)^2}{q_0}\right)=(1+o(1))
(1+\gamma_n^\varepsilon)^{\frac{n(n-1)}{2}},
\]
where $\gamma^\varepsilon=\frac{\kappa_\varepsilon}{n}+\frac{(a_\varepsilon-b_\varepsilon)^2}{4n^2}$, $\kappa_\varepsilon=\frac{(a_\varepsilon-b_\varepsilon)^2}{2(a+b)}$,
and the last equality holds due to the following trivial fact:
\[
\frac{p_{uv}^\varepsilon(\sigma)^2}{p_0}+\frac{q_{uv}^\varepsilon(\sigma)^2}{q_0}=1+\gamma_n^\varepsilon+O(n^{-3}),\,\,\textrm{uniformly for $\sigma\in\{\pm\}^n$.}
\]
Obviously,
$(1+\gamma_n^\varepsilon)^{\frac{n(n-1)}{2}}=\exp\left(\frac{n\kappa_\varepsilon}{2}-
\frac{\kappa_\varepsilon^2}{4}-\frac{\kappa_\varepsilon}{2}+\frac{(a_\varepsilon-b_\varepsilon)^2}{8}\right)$,
hence, $\frac{1}{M}\sum_{l=1}^M g_n^\varepsilon(\sigma[l])=Y_n^\varepsilon+o_P(1)$ if $M\gg\exp\left(\frac{n\kappa_\varepsilon}{2}\right)$.

Next assume $H_1$ holds.
Let $N_{++}=\#\{(u,v): u<v, \sigma_u\sigma_v=+, \tau_u\tau_v=+\}$,
$N_{+-}=\#\{(u,v): u<v, \sigma_u\sigma_v=+, \tau_u\tau_v=-\}$,
$N_{++}=\#\{(u,v): u<v, \sigma_u\sigma_v=-, \tau_u\tau_v=+\}$,
$N_{++}=\#\{(u,v): u<v, \sigma_u\sigma_v=-, \tau_u\tau_v=-\}$.
Similar to the expressions of $s_{r\pm}$ for $r=0,1,2$ in the proof of Theorem \ref{power:case2}, one can derive that
\begin{eqnarray*}
N_{++}&=&\frac{n^2}{8}-\frac{n}{2}+\frac{n}{8}(\rho_1^2+\rho_3^2+\rho_5^2),
N_{+-}=\frac{n^2}{8}+\frac{n}{8}(\rho_1^2-\rho_3^2-\rho_5^2),\\
N_{-+}&=&\frac{n^2}{8}-\frac{n}{8}(\rho_1^2-\rho_3^2+\rho_5^2),
N_{--}=\frac{n^2}{8}-\frac{n}{8}(\rho_1^2+\rho_3^2-\rho_5^2).
\end{eqnarray*}
Following (\ref{def:gamma2+:gamma0-}), one can check that
\begin{eqnarray*}
&&\E_{A,\sigma}g_n^\varepsilon(\sigma)^2\\
&=&4^{-n}\sum_{\sigma,\tau}\prod_{u<v}
\left(\frac{1}{p_0^2}p_{uv}^\varepsilon(\sigma)^2p_{uv}(\tau)+\frac{1}{q_0^2}q_{uv}^\varepsilon(\sigma)^2q_{uv}(\tau)\right)\\
&=&4^{-n}\sum_{\sigma,\tau}\left(\frac{1}{p_0^2}\left(\frac{a_\varepsilon}{n}\right)^2\left(\frac{a}{n}\right)
+\frac{1}{q_0^2}\left(1-\frac{a_\varepsilon}{n}\right)^2\left(1-\frac{a}{n}\right)\right)^{N_{++}}\\
&&\times\left(\frac{1}{p_0^2}\left(\frac{a_\varepsilon}{n}\right)^2\left(\frac{b}{n}\right)
+\frac{1}{q_0^2}\left(1-\frac{a_\varepsilon}{n}\right)^2\left(1-\frac{b}{n}\right)\right)^{N_{+-}}\\
&&\times\left(\frac{1}{p_0^2}\left(\frac{b_\varepsilon}{n}\right)^2\left(\frac{a}{n}\right)
+\frac{1}{q_0^2}\left(1-\frac{b_\varepsilon}{n}\right)^2\left(1-\frac{a}{n}\right)\right)^{N_{-+}}\\
&&\times\left(\frac{1}{p_0^2}\left(\frac{b_\varepsilon}{n}\right)^2\left(\frac{b}{n}\right)
+\frac{1}{q_0^2}\left(1-\frac{b_\varepsilon}{n}\right)^2\left(1-\frac{b}{n}\right)\right)^{N_{--}}\\
&=&(1+o(1))\E_{\sigma\tau}(1+\gamma_{2+})^{N_{++}}(1+\gamma_{2-})^{N_{+-}}(1+\gamma_{0+})^{N_{-+}}
(1+\gamma_{0-})^{N_{--}}.
\end{eqnarray*}
It follows from direct examinations that
\[
(1+\gamma_{2+})^{\frac{n^2}{8}-\frac{n}{2}}(1+\gamma_{2-})^{\frac{n^2}{8}}(1+\gamma_{0+})^{\frac{n^2}{8}}(1+\gamma_{0-})^{\frac{n^2}{8}}
\asymp\exp\left(\frac{n\kappa_\varepsilon}{2}\right),
\]
and
\begin{eqnarray*}
&&\E_{\sigma\tau}(1+\gamma_{2+})^{\frac{n}{8}(\rho_1^2+\rho_3^2+\rho_5^2)}
(1+\gamma_{2-})^{\frac{n}{8}(\rho_1^2-\rho_3^2-\rho_5^2)}
(1+\gamma_{0+})^{-\frac{n}{8}(\rho_1^2-\rho_3^2+\rho_5^2)}
(1+\gamma_{0-})^{-\frac{n}{8}(\rho_1^2+\rho_3^2-\rho_5^2)}\\
&=&(1+o(1))\E_{\sigma\tau}\exp\left(\frac{(a_\varepsilon-b_\varepsilon)^2(a-b)}{4(a+b)^2}\rho_3^2
+\frac{(a_\varepsilon-b_\varepsilon)(a-b)}{2(a+b)}\rho_5^2\right)\\
&\overset{n\to\infty}{\rightarrow}&\left(1-\frac{(a_\varepsilon-b_\varepsilon)^2(a-b)}{2(a+b)^2}\right)^{-1/2}
\left(1-\frac{(a_\varepsilon-b_\varepsilon)(a-b)}{a+b}\right)^{-1/2}.
\end{eqnarray*}
The last limit follows by condition $(a_\varepsilon-b_\varepsilon)(a-b)<a+b$
and asymptotic independent standard normality of $\rho_3$ and $\rho_5$.
Hence, $\E_{A,\sigma}g_n^\varepsilon(\sigma)^2\lesssim
\exp\left(\frac{n\kappa_\varepsilon}{2}\right)$,
leading to $\frac{1}{M}\sum_{l=1}^M g_n^\varepsilon(\sigma[l])=Y_n^\varepsilon+o_P(1)$ if $M\gg\exp\left(\frac{n\kappa_\varepsilon}{2}\right)$.
\end{proof}

\end{document}